\documentclass{mscs}

\usepackage{mathpartir}
\usepackage{xargs}

\usepackage{harvard}
\citationmode{abbr}

\usepackage{authblk}

\usepackage{xparse}

\usepackage[a4paper]{geometry}
\usepackage{nameref,hyperref}
\usepackage{url}
\usepackage[english]{babel}
\usepackage[utf8]{inputenc}
\usepackage{amsmath,stmaryrd,amsfonts,amssymb}
\usepackage{thmtools}

\usepackage[usenames,dvipsnames]{xcolor}
\hypersetup{
 linktocpage,
 colorlinks,
 citecolor=BlueViolet,
 filecolor=red,
 linkcolor=Blue,
 urlcolor=BrickRed
}

\usepackage{tikz-cd}

\newenvironment{diagram}{\begin{tikzcd}[sep=large]}{\end{tikzcd}}
\newenvironment{largediagram}{\begin{tikzcd}[sep=huge]}{\end{tikzcd}}

\newenvironment{proofof}[1]
{\begin{proof}[Proof of {#1}]}
{\end{proof}}
\newenvironment{proofsketch}
{\begin{proof}[Proof (sketch)]}
{\end{proof}}

\DeclareSymbolFont{sfoperators}{OT1}{cmss}{m}{n}
\DeclareSymbolFontAlphabet{\mathsf}{sfoperators}

\newtheorem{lemma}{Lemma}[section]
\newtheorem{proposition}{Proposition}[section]
\newtheorem{corollary}{Corollary}[section]
\newtheorem{theorem}{Theorem}[section]

\newtheorem{definition}{Definition}[section]

\makeatletter
\def\operator@font{\mathgroup\symsfoperators}
\makeatother

\DeclareDocumentCommand{\hom}{o m m}{\ensuremath{\operatorname{Hom}_{\IfValueT{#1}{#1}}\left(#2,#3\right)}}

\DeclareDocumentCommand{\triple}{m m o o}{
  \IfValueTF{#3}
  {\ensuremath{#1 \overset{#2}{\longrightarrow} #3 \overset{#4}{\longrightarrow} \NN}}
  {\ensuremath{#1 \overset{#2}{\longrightarrow} \Delta_{#2} \overset{\delta_{#2}}{\longrightarrow} \NN}}}

\newcommand{\isetsep}{\;\ifnum\currentgrouptype=16 \middle\fi|\;}
\newcommand{\id}[1]{\ensuremath{\text{id}_{#1}}}
\newcommand{\stream}{\ensuremath{\operatorname{Str}}}
\newcommand{\gstream}[1]{\ensuremath{\operatorname{Str^{#1}}}}
\newcommand{\ghead}[1]{\ensuremath{\operatorname{head^{#1}}}}
\newcommand{\gtail}[1]{\ensuremath{\operatorname{tail^{#1}}}}
\newcommand{\tail}{\ensuremath{\operatorname{tail}}}
\newcommand{\nxt}[1]{\ensuremath{\operatorname{next}^{#1}}}
\newcommand{\NN}{\ensuremath{\mathbb{N}}}

\newcommand{\CC}{\ensuremath{\mathbb{C}}}

\newcommand{\DD}{\ensuremath{\mathbb{D}}}

\newcommand{\Cl}{\ensuremath{\mathcal{C}}}

\newcommand{\cwffont}[1]{\mathrm{#1}}

\newcommand{\cat}[1]{\mathcal{#1}}
\newcommand{\Fam}[2][\BaseCat]{#1(#2)}
\newcommand{\ElCwF}[3][\BaseCat]{#1(#2 \vdash #3)}
\newcommand{\p}[0]{\cwffont{p}}
\newcommand{\q}[0]{\cwffont{q}}
\newcommand{\compr}[2]{#1.#2}
\newcommand{\cpair}[2]{\pair{#1}{#2}}

\newcommandx{\cpairmore}[3][1={\kappa_1},2={\kappa_n}]{\left\langle{#1}, \dots, {#2}, #3\right\rangle}

\newcommand{\PiSem}[2]{\Pi(#1, #2)}
\newcommand{\SigmaSem}[2]{\Sigma(#1, #2)}

\newcommand{\ev}{\cwffont{ev}}
\newcommand{\clockfam}{\cwffont{clock}}
\newcommand{\laterfam}[1]{\later{#1}}
\newcommand{\nextel}[1]{\cwffont{next}^{#1}}
\newcommand{\fixel}[2]{\cwffont{fix}^{#1}(#2)}
\newcommand{\inel}[2]{\cwffont{in}_{#1,#2}}

\newcommand{\lambdael}[1]{\lambda(#1)}
\newcommand{\idfam}[3]{\cwffont{Id}_{#1}(#2,#3)}

\newcommand{\delayedEl}[3]{\BaseCat(#1 \rightarrowtriangle^{#3}#2)}

\newcommand{\inv}[1]{#1^{-1}}

\newcommand{\Clproj}[2]{\pi_{#1, #2}}

\newcommand{\Clin}[2]{\cwffont{in}_{#1, #2}}

\newcommand{\Piel}[3]{\code{\Pi^{#1}}(#2, #3)}
\newcommand{\Sigmael}[3]{\code{\Sigma^{#1}}(#2, #3)}

\newcommand{\forallel}[2]{\code{\forall^{#1}}(\lambdael{#2})}

\newcommand{\laterel}[1]{\code{\later{#1}}}

\newcommand{\gr}[1]{\ensuremath{\mathsf{GR}(#1)}}
\newcommand{\grold}[1]{\ensuremath{\mathfrak{G}\mathfrak{R}\left(#1\right)}}

\newcommand{\alwaysapp}[2]{\ensuremath{#1\!\left[#2\right]}} 

\newcommand{\laterbare}{\ensuremath{\operatorname{\blacktriangleright}}}
\DeclareDocumentCommand{\later}{m o m}{
  \IfNoValueTF{#2}
  {\operatorname{\overset{#1}\laterbare}#3}
  {\operatorname{\overset{#1}\laterbare}#2 . #3}}

\newcommand{\dsubst}[4]{\ensuremath{#2 : #3\rightarrowtriangle^{#1} #4}}
\newcommand{\hrt}[1]{\left[ #1 \right]}

\newcommand{\term}[1]{\ensuremath{\operatorname{\mathsf{#1}}}}

\DeclareDocumentCommand{\latercode}{m m}{
  {\operatorname{\overline\laterbare^{#1}} #2}}
\newcommand{\forallcode}{\ensuremath{\overline{\forall}}}
\newcommand{\delayapp}[1]{\ensuremath{\mathbin{\circledast^{#1}}}}

\DeclareDocumentCommand{\next}{ m o m }{
  \IfNoValueTF{#2}
  {\term{next}^{#1}#3}
  {\term{next}^{#1} #2 . #3}}

\newcommand{\prev}[1]{\ensuremath{\term{prev}#1}}
\newcommand{\force}{\ensuremath{\term{force}}}
\newcommand{\depprod}[3]{\ensuremath{{\textstyle\prod\left(#1 : #2\right) . #3}}}
\newcommand{\depsum}[3]{\ensuremath{{\textstyle\sum\left(#1 : #2\right) . #3}}}
\newcommand{\idty}[3]{\ensuremath{{\mathsf{Id}_{#1}\left(#2, #3\right)}}}
\newcommand{\refl}[2]{\ensuremath{{\term{refl}_{#1}#2}}}
\newcommand{\depprodcode}[4]{\ensuremath{{\textstyle\overline{\prod}_{#1}\left(#2 : #3\right) . #4}}}
\newcommand{\depsumcode}[4]{\ensuremath{{\textstyle\overline{\sum}_{#1}\left(#2 : #3\right) . #4}}}

\newcommand{\Nat}{\ensuremath{\mathbf{N}}}

\newcommand{\timescode}{\overline\times}
\newcommand{\gstreamcode}[2]{\ensuremath{\overline{\operatorname{Str}}^{#1}}(#2)}

\newcommand{\fv}{\mathsf{fv}}

\newcommand{\univ}[1]{\ensuremath{\operatorname{U}_{#1}}}
\newcommand{\elems}[1]{\ensuremath{\operatorname{El}_{#1}}}
\newcommand{\univin}[2]{\term{in}_{#1,#2}}

\newcommand{\clocktype}{\ensuremath{\mathsf{clock}}}

\newcommand{\pair}[2]{\ensuremath{\left\langle #1, #2 \right\rangle}}

\newcommand{\proj}[2]{\ensuremath{\pi_{#1} #2}}
\newcommand{\fix}[1]{\ensuremath{\term{fix}^{#1}}}
\newcommand{\adv}[1]{\ensuremath{\term{adv}^{#1}}}

\newcommand{\subst}[2]{\ensuremath{\left[#2/#1\right]}}

\newcommand{\wfctx}[1]{\ensuremath{#1 \vdash \,}} 
\newcommand{\wfclock}[2]{\ensuremath{#1 \vdash #2 : \clocktype}}
\newcommand{\wftype}[2]{\ensuremath{#1 \vdash #2 \, \operatorname{type}}}
\newcommand{\hastype}[3]{\ensuremath{#1 \vdash #2 : #3}}
\newcommand{\judgeeq}[3]{\ensuremath{#1 \vdash #2 \jeq #3}}
\DeclareDocumentCommand{\eqjudg}{ m m m o }{
    \IfNoValueTF{#4}
    {\ensuremath{#1 \vdash #2 \jeq #3}}
    {\ensuremath{#1 \vdash #2 \jeq #3 : #4}}
}

\newcommand{\defeq}{\eqdef}
\newcommand{\empctx}{\cdot}

\newcommand{\comp}{\mathrel{\circ}}

\newcommand{\jeq}{=}

\newcommand{\iso}{\cong}

\newcommand{\gdtt}{\textsf{GDTT}}
\newcommand{\coregdtt}{\textsf{Core-GDTT}}

\newcommand{\clocks}{\ensuremath{\mathcal{C}}}
\newcommand{\pindex}[1]{\mathfrak{I}(#1)}

\newcommand{\TimeCat}{\ensuremath{\mathbb{T}}}

\newcommand{\sets}{\ensuremath{\mathrm{Set}}}

\newcommand{\minkappa}[2]{\ensuremath{#1^{-#2}}}
\newcommand{\eps}{\varepsilon}
\newcommand{\colim}{\ensuremath{\varinjlim}}

\newcommand{\code}[1]{\ulcorner #1 \urcorner}

\newcommand{\Elsem}[1]{\mathcal{E}l^{#1}}
\newcommand{\Usem}[1]{\mathcal{U}^{#1}}

\newcommand{\ElfamOne}[1]{\Elsem{#1}_1}
\newcommand{\ElfamTwo}[1]{\Elsem{#1}_2}

\newcommand{\Ufam}[1]{\Usem{#1}_1}

\newcommand{\pullbacktip}{\arrow[dr, phantom, "\scalebox{2}{\ensuremath{\lrcorner}}", very near start]}

\newcommand{\eqdef}{\overset{\textrm{def}}{=}}

\newcommand{\catfont}{\mathbb}

\newcommand{\ccat}{{\catfont{C}}}

\newcommand{\opcat}[1]{{{#1}^{\mathrm{op}}}}

\newcommand{\Fin}{\mathbf{Fin}}

\newcommand{\BaseCat}{\sets^{\TimeCat}}

\newcommand{\GAP}{\hspace{1cm}}

\newcommand{\Hom}{\ensuremath{\operatorname{Hom}}}

\newcommand{\den}[1]{\ensuremath{\left\llbracket #1 \right\rrbracket}}

\newcommand{\uniquemap}{!}

\newcommand{\restrict}[2]{{#1}|_{#2}}
\newcommand{\cores}[2]{\overline{#1}^{#2}}

\newcommand{\FSA}{\mathcal{E}}
\newcommand{\FSB}{\mathcal{E'}}
\newcommand{\FSC}{\mathcal{E''}}

\newcommand{\tick}[2]{\mathrm{tick}^{#1}}
\newcommand{\fresh}[1]{\lambda_{#1}}

\newcommand{\clocksetA}{\chi}
\newcommand{\clocksetB}{\chi'}
\newcommand{\clocksetC}{\chi''}

\newcommand{\clocksetAmap}{\langle\chi\rangle}
\newcommand{\clocksetBmap}{\langle\chi'\rangle}

\newcommand{\nattype}{\term{Nat}}

\title{Denotational semantics for guarded dependent type theory}
\author[Ale\v{s} Bizjak and Rasmus Ejlers M{\o}gelberg]
{A\ls L\ls E\ls \v{S}\ns B\ls I\ls Z\ls J\ls A\ls K$^1$%
\ns and\ns%
R\ls A\ls S\ls M\ls U\ls S\ns E\ls J\ls L\ls E\ls R\ls S\ns M\ls{\O}\ls G\ls E\ls L\ls B\ls E\ls R\ls G$^2$%
 \thanks{Corresponding author. Full address: IT University of Copenhagen, Department of Computer Science, Rued Langgaards Vej 7,
 2300 Copenhagen, Denmark. Email:  \href{mailto:mogel@itu.dk}{mogel@itu.dk}}
\addressbreak
$^1$ Aarhus University%
\addressbreak%
$^2$ IT University of Copenhagen%
}

\begin{document}

\maketitle

\begin{abstract}
  We present a new model of Guarded Dependent Type Theory ($\gdtt$), a type theory with guarded recursion and multiple 
  clocks in which one can program with, and reason about coinductive types. Productivity of recursively defined coinductive 
  programs and proofs is encoded in types using guarded recursion, and can therefore be checked modularly, unlike
  the syntactic checks implemented in modern proof assistants. 
  
  The model is based on a category of covariant presheaves over a category of time objects, and quantification over clocks
  is modelled using a presheaf of clocks. To model the \emph{clock irrelevance axiom}, crucial for programming with
  coinductive types, types 
  must be interpreted as presheaves internally right orthogonal to the object of clocks. In the case of dependent types, this translates to 
  a lifting condition similar to the one found in homotopy theoretic models of type theory, but here with
  an additional requirement of uniqueness of lifts. Since the universes 
  defined by the standard Hofmann-Streicher construction in this model do not satisfy this property, the universes in 
  $\gdtt$ must be indexed by contexts of clock variables. We show how to model these universes in such a way that
  inclusions of clock contexts give rise to inclusions of universes commuting with type operations on the nose.
\end{abstract}

\section{Introduction}
\label{sec:introduction}

Type theories with dependent types such as \possessivecite{MartinLof:73} Type Theory 
or the Extended Calculus of Constructions~\cite{Luo:94} 
are systems that can be simultaneously 
thought of as programming languages and logical systems. One reason why this is useful is that programs, their
specification and the proof that a program satisfies this specification, can be expressed in the same language. 
In these systems, the logical interpretation of terms forces a totality requirement on the programming language, 
i.e., rules out general
recursion, since nonterminating programs can inhabit any type, and thus be interpreted as proofs of false statements.

The lack of general recursion is a limitation both from a programming and a logical perspective. 
For example, when programming with coinductive types, the natural way to program and reason about these is by recursion. 
For example, the constant stream of zeros can be naturally described as the solution to the equation
$\term{zeros} = 0 :: \term{zeros}$. 
To ensure logical consistency, such recursive definitions must be productive, in the sense that any finite segment of 
the stream can be computed in finite time. 
Modern proof assistants such as Coq~\citeyear{Coq:manual} and Agda~\cite{Norell:thesis} do support coinductive types and recursive definitions such as the 
above but the productivity checks are based on a syntactical analysis of terms, and are not modular. This means that
using these in larger applications requires sophisticated tricks~\cite{NAD:beat}. This paper is concerned with a new technique
using guarded recursion to express productivity in types. 

Guarded recursion in the sense of \citeasnoun{Nakano:modality} 
is a safe way of adding recursion to type theory without breaking logical 
consistency. The idea is to guard all unfoldings of recursive equations by time steps in the form of a modal type 
constructor $\later{}$. 
The type $\later{}A$ should be thought of as a type of elements of $A$ available one time step from now. Values can be
preserved by time steps using an operator $\nxt{}$
satisfying $\nxt{} t : \later{} A$ whenever $t : A$. The fixed point operator has type
$\fix{} : (\later{} A \to A) \to A$ and computes, for any $f$, a fixed point for $f \circ \nxt{}$. 
This is particularly useful when programming
with \emph{guarded recursive types}, i.e., recursive types where all occurrences of the type parameter appears guarded
by a $\later{}$. For example, a guarded recursive type of streams would satisfy 
$\gstream{} =  \nattype \times \later{}\gstream{}$ and the stream of zeros can be defined as 
$\fix{}(\lambda xs. \pair{0}{xs})$.  The type 
$\later{} \gstream{} \to \gstream{}$ in fact exactly captures productive recursive stream definitions. 
Using universes, the type $\gstream{}$ can itself be computed as a guarded recursive fixed
point. In this paper we use universes \`a la Tarski, i.e., for any term $A : \univ{}$ there is a type $\elems{}(A)$. If 
we assume an operation $\latercode{} : \later{}U \to U$ satisfying $\elems{}(\latercode{}(\nxt{}(A))) = \later{}\elems{}(A)$,
then the type of guarded streams can be encoded as 
$\gstream{} \eqdef\elems{}(\fix{}(\lambda X. \nattype \times \latercode{}(X)))$. 

The guarded recursive type of streams above is not the usual type of streams. 
In particular, a term of type $\gstream{} \to \gstream{}$ must always
be \emph{causal} in the sense that the $n$ first element of output only depend on the $n$ first elements of input. Indeed, 
causality of maps is crucial for the encoding of productivity in types. 
On the other hand, a \emph{closed} term of type $\gstream{}$ does denote a full stream of numbers,
and likewise a term of type $\gstream{}$ in a context consisting solely of a variable $x : \nattype$ gives rise to an assignment
of numbers to full streams of numbers. In general, this holds if the context is stable, i.e., consists entirely of time-independent 
types. 

\subsection{Guarded recursion with multiple clocks}

\citeasnoun{Atkey:Productive} proposed a way to program with 
coinductive types using this idea, expressing time-independence
by indexing all $\later{}$ operators, $\nxt{}$ and $\fix{}$ by clocks. For example, if $t : A$ and $\kappa$ is a clock, 
$\nxt{\kappa} t : \later\kappa A$ and the type $\later\kappa A$ is to be thought of as elements of type $A$ available
one $\kappa$-time step from now. Likewise the guarded recursive type of streams must be indexed with a clock and 
assumed to satisfy $\gstream{\kappa} =  \nattype \times \later{\kappa}\gstream{\kappa}$. 
There are no operations on clocks, only clock variables, although we will
see that a single clock constant can be useful. We refer to this as \emph{guarded recursion with multiple clocks}, and 
the case of a single operator $\later{}$ as \emph{guarded recursion with a single clock} or sometimes simply 
\emph{the single clock case}. 

In turn, \emph{clock quantification} of guarded dependent type theory allows us to define the coinductive type of streams from the guarded recursive type of streams as $\forall\kappa .  \gstream\kappa$.
Clock quantification behaves similarly to the dependent product type in the sense it has analogous introduction and elimination rules; terms of this type are introduced by clock abstraction $\Lambda \kappa . t$, and eliminated using clock application $t[\kappa']$, provided $\kappa'$ is a valid clock.
However, clock quantification additionally satisfies the \emph{clock irrelevance} property, which is crucial for showing that types such as $\forall\kappa .  \gstream\kappa$ satisfy the properties expected of coinductive types, i.e., that they are final coalgebras.
Using these constructs and properties we can program with streams using guarded recursion, ensuring productivity of definitions using types.


This paper presents a model of $\gdtt$~\cite{Bizjak-et-al:GDTT}, an extensional type theory with
guarded recursion and clocks, in which one can program with, and reason about guarded recursive and coinductive
types. To motivate some of the constructions of $\gdtt$, we now take a closer look at the encoding of coinductive streams
as  $\forall\kappa .  \gstream\kappa$. As a minimal requirement for this to work, we need an isomorphism of types  
$\forall\kappa .  \gstream\kappa \iso \nattype \times \forall\kappa . \gstream\kappa$. 
This isomorphism is a composition of three isomorphisms
\begin{align*}
 \forall\kappa .  \gstream\kappa & = \forall\kappa .  \nattype \times \later\kappa\gstream\kappa \\
 & \iso (\forall\kappa .  \nattype) \times \forall\kappa . \later\kappa\gstream\kappa \\
 & \iso \nattype \times \forall\kappa . \later\kappa\gstream\kappa \\
 & \iso  \nattype \times \forall\kappa . \gstream\kappa 
\end{align*}
The first isomorphism follows from the fact that $\forall\kappa.A$ behaves essentially as the dependent product type $\depprod{\kappa}{\clocktype}A$, and thus distributes over binary products.
For the second isomorphism, we need $\nattype\iso \forall\kappa.\nattype$.
One direction of this isomorphism maps $x: \nattype$ to $\Lambda\kappa .
x$, and the opposite way evaluates an element in $\depprod{\kappa}{\clocktype}\nattype$ at a \emph{clock constant} $\kappa_0$.
The composition on $\nattype$ is obviously the identity, but for the other composition to be the identity, we need to assume the $\eta$-axiom for $\forall\kappa .
A$, and the \emph{clock irrelevance axiom}, which states that whenever $t : \forall\kappa .
A$ and $\kappa$ is not in $A$, then evaluating $t$ at different clocks give the same result.
One of the main contributions of this paper is that this axiom can be modelled using a notion of orthogonality.
The last isomorphism requires an inverse $\force : \forall\kappa.
\later\kappa A \to \forall\kappa .
A$ to the map induced by $\nxt\kappa$.

In this paper we focus on modelling $\gdtt$, and refer the reader to~\cite{Mogelberg:tt-productive-coprogramming} 
for a proof of correctness of the coinductive type encodings. 

\subsection{A model of guarded recursion with multiple clocks}

In the single clock case guarded recursion can be modelled in the topos of trees, i.e., the category $\sets^{\opcat{\omega}}$ 
of presheaves over the ordered natural numbers $\omega$. In this model, a closed type is modelled as a sequence of sets
$(X_n)_{n \in \NN}$ together with restriction maps $X_{n+1} \to X_n$. We think of $X_n$ as the type as
it looks if we have $n$ steps to reason about it. For example, in the guarded recursive type of streams, since the tail takes
one computation step to compute, one can compute the $n+1$ first elements of the stream in $n$ steps. We can represent
this by the object defined as $\gstream{}_n = \NN^{n+1}$ with restriction maps as projections. 

In this model $\later{} X$ is the object given by $(\later{}X)_0 = 1$ and $(\later{}X)_{n+1} = X_n$. Redefining 
$\gstream{}_n$ to be  $\NN^{n+1} \times 1$ (and associating products to the right) one gets 
$\gstream{} = \NN \times \later{}\gstream{}$. In the empty context a term $t : \later\kappa A \to A$ is modelled as a 
family of maps $t_{n+1} : A_n \to A_{n+1}$ and $t_0 : 1 \to A_0$. The fixed point operator maps such a family to the 
global element $\fix{}(t) : 1 \to A$ defined as $\fix{}(t)_n = t_n\circ \dots \circ t_0$. We refer 
to~\cite{Birkedal+:topos-of-trees} for further details. 

In this paper we extend this to a model of guarded recursion with multiple clocks. 
The model is a presheaf category over a category $\TimeCat$ 
of time
objects. In the single clock case, a time object was simply a number indicating the number of ticks left on the unique clock.
In the case of multiple clocks, a time object consists of a finite set of clocks $\FSA$, together with a map $\delta : \FSA \to \NN$
indicating the number of ticks left on each clock. A morphism of time objects $\sigma : (\FSA, \delta) \to (\FSB, \delta')$ is 
a map $\sigma : \FSA \to \FSB$ such that $\delta'(\sigma(\lambda)) \leq \delta(\lambda)$ for each $\lambda\in \FSA$. 
Such a morphism can rename clocks, introduce new clocks (elements of $\FSB$ outside the image of $\sigma$) 
and even synchronise clocks (by mapping them to the same clock). The inequality requirement corresponds to the inequalities
between numbers in the topos of trees. 

We consider \emph{covariant} presheaves on $\TimeCat$, i.e., the category of functors $\TimeCat \to \sets$. In this category
there is an object of clocks given by $ \Cl(\FSA,\delta) = \FSA$, which we use to model clock variables. Clock quantification is
modelled as a dependent product over $\Cl$. With this interpretation, for a type $A$ in which $\kappa$ does not appear free,
the type $\forall\kappa . A$ is modelled as a simple function type $\Cl\to A$. The clock irrelevance axiom mentioned above 
then states that the map $A \to (\Cl\to A)$ mapping an element $x$ in $A$ to the constant map to $x$ is an isomorphism. Of 
course, this does not hold for all presheaves $A$, and so we must show that this holds for the interpretation of any type. Note
that it does not hold for $A = \Cl$, and so, although $\forall\kappa . A$ is modelled as a dependent product over the 
presheaf of clocks, there is no type of clocks in the type theory. This is similar to the status of the interval in cubical 
type theory~\cite{Cubical},
which is not itself a type, but still the set of types is closed under dependent products over the interval (these are path types). 

For dependent types the condition becomes a unique lifting property. In a presheaf model of type theory
a type depending on a context is modelled as a family $A$ over a presheaf $\Gamma$. To this can be associated
a projection $\p : \compr\Gamma A \to \Gamma$ corresponding to syntactic projection between contexts. 
This must satisfy the condition that for all $Y$, and for all
commutative squares as in the outer square below (where $\pi_Y$ is the projection), 
there exists a unique $h$ such that the two triangles commute. 
\[
  \begin{diagram}
    Y \times \Cl \ar{r}{f} \ar{d}{\pi_Y} & \compr\Gamma A\ar{d}{p}\\
    Y \ar{r}{g} \ar[dotted]{ur}[description]{h} & \Gamma
  \end{diagram}
\]
We say that such a map $p$ is \emph{internally right orthogonal} to $\Cl$. 
This condition is similar to the notion of fibration used in models of homotopy theoretic models of type 
theory~\cite{awodey2009homotopy,Simplicial:model} and cubical type theory~\cite{BezemCH13}, 
except that here the liftings are unique. This means that it can be considered a property
that must be proved for each type, rather than structure that is part of the interpretation of a type.

\subsection{Universes}

Since our model is a presheaf category, one would hope that modelling universes would follow the standard 
Hofmann-Streicher construction~\cite{Hofmann-Streicher:lifting}, 
restricting to the elements internally right orthogonal to $\Cl$. Unfortunately, this universe $\Usem{}$ 
is not itself internally right orthogonal to $\Cl$. The reason is that there is a map $\Usem{}\times \Cl \to \Usem{}$ mapping a type $A$ and 
a clock $\kappa$ to $\later\kappa A$, and this map is not constant in the $\Cl$ component. This is a new semantic 
manifestation of a known problem, and we follow the solution used in $\gdtt$, which is to have a family of universes 
$(\univ{\Delta})_\Delta$ in the syntax, indexed by  finite sets of clock variables. 
Each universe $\univ\Delta$ is to be thought of as
the universe of types independent of the clocks outside of $\Delta$, and the type operation $\later\kappa$ is 
restricted on the universe $\univ\Delta$ to the $\kappa$ in $\Delta$. 

This means that universes are indexed by a new dimension, similar to the indexing of universes by natural numbers
used to avoid Russell's paradox~\cite{MartinLof:73}.
Fortunately, there are inclusions $\univ\Delta \to \univ{\Delta'}$ for $\Delta\subseteq\Delta'$, and we prove 
universe polymorphism in this dimension. This means that 
operations on types such as dependent product can be defined on the universes in such a way that they commute with
the inclusions mentioned above, not just up to isomorphism, but indeed up to identity. 
We hope that, as a consequence of this result, 
the indexing of universes by clock contexts can be suppressed in practical applications, just like
the indexing by natural numbers is often suppressed.

\subsection{Related work}

The notion of guarded recursion studied in this paper originates with \citeasnoun{Nakano:modality}. Much of 
the recent interest in guarded recursion is due to the guarded recursive types, which can even have negative occurences
and thus, by adding $\later{}$ operators in appropriate places, provide approximations to solutions 
to equations that can not be solved in set theory. These have 
been used to construct 
syntactic models and operational reasoning principles for (also
combinations of) advanced programming language features including
general references, recursive types, countable non-determinism and
concurrency~\cite{Birkedal+:topos-of-trees,Bizjak-et-al:countable-nondet-internal,BirkedalL:icap}. 
This technique can be understood as an 
abstract form of step-indexing~\cite{Appel:M01}, the connection to which was first discovered by 
\citeasnoun{AppelMRV07}. Most of these applications
have been constructed using logics with guarded recursion, such as the internal language of the topos of 
trees~\cite{Birkedal+:topos-of-trees}, 
but recently $\gdtt$ has been used to construct denotational models of programming languages like 
FPC~\cite{Paviotti:LICS:2016}, 
modelling the recursive types of these as guarded recursive types. 

Most type theories with guarded recursion considered until
now have been extensional, with the exception of \emph{guarded cubical type theory}~\cite{gctt}. This has, however, only 
been developed in the single clock case, although there exists an experimental version with multiple clocks.

Guarded recursion with multiple clocks was first developed in the simply typed setting by \citeasnoun{Atkey:Productive}. 
The second named author~\cite{Mogelberg:tt-productive-coprogramming}
extended these results to a model of dependent type theory and proved correctness of the coinductive type encodings 
inside a type theory with guarded recursion.  These two early works used a restricted version of clock application, allowing
$\alwaysapp t{\kappa'} : A\subst\kappa{\kappa'}$ for $t : \forall\kappa . A$ only if $\kappa'$ does not appear free in 
$\forall\kappa . A$. This condition can be thought of as disallowing the clocks $\kappa$ and $\kappa'$ to be synchronised
in $A$, and was motivated by the models considered at the time. This restriction has unfortunate consequences for 
the syntactic metatheory. In particular, the present authors do not know how to prove type preservation for 
clock $\beta$-reductions in these systems. 

This led us to suggest a different model~\cite{Bizjak-Moegelberg:clocks-model} given by a family of presheaf categories $\grold{\Delta}$ indexed by clock contexts
(finite sets of clock variables) $\Delta$. This model should in principle lead to a model of $\gdtt$, but this was never done in
detail, due to a problem with modelling substitution of clock variables. Such substitutions are given by maps 
$\sigma : \Delta \to \Delta'$ and must correspond semantically to functors $\grold{\Delta} \to \grold{\Delta'}$. 
While these functors can be defined in a natural 
way, they do not commute with dependent function types up to identity, only up to isomorphism. This problem can be thought
of as a coherence problem, similar to the one arising when modelling type theory in locally cartesian closed 
categories~\cite{Hofmann:lccc-strictification}. 
It is very likely that \possessivecite{Hofmann:lccc-strictification} solution to the latter problem can be adapted to 
construct an equivalent family of categories for which the functors preserve construction on the nose,
but we prefer the solution presented here, 
which organises all these categories inside one big presheaf category, thereby reducing
the model construction to the known construction of modelling type theory in a presheaf category. The precise relation to the
categories $\grold{\Delta}$ is discussed in Section~\ref{sec:discussion}.

Recently, $\gdtt$ has been refined to \emph{clocked type theory} (CloTT)~\cite{clott}, which has better operational
properties, and indeed strong normalisation has been proved for clocked type theory in the setting without identity types. 
The principal novel feature of CloTT is the notion of ticks on a clock introduced in contexts as 
assumptions of the form $\alpha  : \kappa$, for $\kappa$ a clock. Ticks can be used to 
encode the delayed substitutions (see Section~\ref{sec:delayed:subst}) of $\gdtt$, 
and reduce most of the equalities between these to  
$\beta$ and $\eta$ equalities. Since the initial development of the research reported here, 
\citeasnoun{CloTTmodel} have developed a model CloTT based on the model presented here. Their
paper however, does not describe how to model the clock irrelevance axiom, nor universes
as presented here. Also, the presence of ticks makes the model construction for CloTT rather complicated
and so we have chosen to present the model in the simpler setting of $\gdtt$ first.


In recent work on guarded computational type theory, \citeasnoun{sterling-harper:2018}
propose a clock intersection connective to be used as a special `irrelevant' quantification over clocks.
Using this they encode coinductive
types, while avoiding 
the indexing of universes by clock contexts as done here. Irrelevant clock quantification is interpreted using intersection of sets
in a syntactic model, in which types are essentially indexed sets of values. This is similar to the original
interpretation of clock quantification in the work of \citeasnoun{Atkey:Productive}. A related irrelevant quantification over sizes 
appears in the work of \citeasnoun{Abel:NBE:sized:types}. However, it is unclear how to give \emph{denotational} semantics
of such a constructor. We remark that the model used by \citeasnoun{sterling-harper:2018}
is based on a category very similar to the presheaf category used in this paper and that these models were discovered independently.

One way of understanding the need for multiple clocks for encoding coinductive types is that they provide a controlled
way of eliminating the $\later{}$ modality as in the term $\force : \forall\kappa. \later\kappa A \to \forall\kappa . A$ mentioned above. 
As an alternative solution to this problem, \citeasnoun{birkedal2017guarded} have 
suggested to use an \emph{always} modality $\blacksquare$ satisfying $\blacksquare \later{}A \iso \blacksquare A$. It is yet
unclear how far this idea can be extended, in particular if it can be used for encoding nested inductive and coinductive types. 

Sized types \cite{HughesPS96} offer a different approach to the problem of encoding productivity
in types. The idea is to annotate approximations of a coinductive type with the number of unfoldings that can be applied
to it. The real coinductive type is then the approximation associated with an infinite ordinal. 
When programming with sized types, the sizes sometimes get in the way, 
motivating the concept of irrelevant quantification over sizes mentioned above. The syntactic theory of
sized types is further developed than that of guarded recursion \cite{Abel:Wellfounded,Abel:NBE:sized:types,Sacchini13}, 
and sized types are also available in an experimental extension of Agda. Sized types have not been used as abstract
step-indexing in the sense described above for guarded recursion, and the authors are not aware of any work on 
denotational semantics for sized types. 

Our view is that guarded recursion should be thought of as an abstraction of sized types, providing similar benefits
as the abstraction of step-indexing, in particular by hiding Kripke structure present in the model. 
This view is supported by work by \citeasnoun{Veltri:19} in which a model of guarded recursion is constructed 
in Agda using sized types to model recursion. In this work, the model is restricted to a simply typed language 
language specialised to the case of just  $0$ or $1$ clocks, thus avoiding the issue of clock synchronisation 
treated in this paper.

\subsection{Overview}

Section~\ref{sec:basic:syntax} presents a basic type theory $\coregdtt$ for guarded recursion with multiple clocks. This 
can be thought of as the core of $\gdtt$~\cite{Bizjak-et-al:GDTT} although we use a slightly different presentation. 
Section~\ref{sec:basics}
then presents a basic model of $\coregdtt$ in the presheaf category $\sets^\TimeCat$, and Section~\ref{sec:orthogonality} 
shows how to 
model the clock irrelevance axiom. The following sections~\ref{sec:identity:types} and~\ref{sec:delayed:subst} 
then extend $\coregdtt$ with extensional identity
types and delayed substitutions, a construction needed for reasoning about guarded recursive and coinductive types. 
Section~\ref{sec:universes} is devoted to universes and modelling universe polymorphism in the 
clock context dimension and Section~\ref{sec:interp:syntax}
sketches how to extend Hofmann's interpretation of dependent type theory syntax~\cite{Hofmann:syntax-and-semantics}
to interpreting \gdtt\ into the model presented in this paper.
Finally the relations to the categories $\grold\Delta$ constructed in previous 
work~\cite{Bizjak-Moegelberg:clocks-model}
by the authors is discussed in Section~\ref{sec:discussion}. 

\section{A basic type theory for guarded recursion}
\label{sec:basic:syntax}

This section introduces $\coregdtt$ a presentational variant of a fragment of the type theory $\gdtt$~\cite{Bizjak-et-al:GDTT}.
The fragment is the one not mentioning universes, delayed substitutions and identity types.
All these will be treated in Sections~\ref{sec:identity:types}--\ref{sec:universes}.
The variation referred to above is in the treatment of clocks, which in previous work~\cite{Bizjak-et-al:GDTT,Mogelberg:tt-productive-coprogramming,Bizjak-Moegelberg:clocks-model} had a separate context.
Here we simply include them in the context as if they were ordinary variables to simplify the presentation of the 
denotational semantics. 
Section~\ref{sec:variation:syntax} sketches an equivalence between $\coregdtt$ and the corresponding 
fragment of $\gdtt$. 

\begin{figure}[tbp]
\paragraph{Wellformed contexts}
\begin{mathpar}
\inferrule*{\,}{\wfctx{\empctx}{}} \and
\inferrule*{\wftype{\Gamma}{A} \\ x\notin\Gamma}{\wfctx{\Gamma, x : A}{}} \and
\inferrule*{\wfctx{\Gamma} \\ \kappa\notin\Gamma}{\wfctx{\Gamma, \kappa : \clocktype}{}} 
\end{mathpar}

\paragraph{Wellformed clocks}
\[
\inferrule*{ \kappa : \clocktype \in\Gamma}{\wfclock{\Gamma}{\kappa}}
\]

\paragraph{Type formation}

\begin{mathpar}
\inferrule*{\wftype{\Gamma, x : A}{B}}{\wftype{\Gamma}{\depprod xAB}} \and
\inferrule*{\wftype{\Gamma, x : A}{B}}{\wftype{\Gamma}{\depsum xAB}} \and
\inferrule*{\wftype{\Gamma}{A} \\ \wfclock{\Gamma}\kappa}{\wftype{\Gamma}{\later\kappa A}} \and
\inferrule*{\wftype{\Gamma, \kappa : \clocktype}{A}}{\wftype{\Gamma}{\forall \kappa . A}} 
\end{mathpar}
\paragraph{Typing judgements}
\begin{mathpar}
\inferrule*{\,}{\hastype{\Gamma, x : A, \Gamma'}x{A}} \and
\inferrule*{\hastype{\Gamma, x : A}t{B}}{\hastype{\Gamma}{\lambda x . t}{\depprod xAB}} \and
\inferrule*{\hastype{\Gamma}{t}{\depprod xAB} \\ \hastype{\Gamma}{u}{A}}{\hastype{\Gamma}{t\, u}{B\subst xu}} \and
\inferrule*{\hastype{\Gamma}t{A} \\ \hastype{\Gamma}u{B\subst xt}}{\hastype{\Gamma}{\pair tu}{\depsum xAB}} \and
\inferrule*{\hastype{\Gamma}{t}{\depsum xAB}}{\hastype{\Gamma}{\proj 1t}{A}} \and
\inferrule*{\hastype{\Gamma}{t}{\depsum xAB}}{\hastype{\Gamma}{\proj 2t}{B\subst x{\proj 1t}}} \and
\inferrule*{\hastype{\Gamma}t{A} \\ \wfclock{\Gamma}{\kappa}}{\hastype{\Gamma}{\nxt\kappa t}{\later\kappa A}} \and  
\inferrule*{\hastype{\Gamma, \kappa : \clocktype}t{\later\kappa A}}
{\hastype{\Gamma}{\prev{\kappa} . t}{\forall\kappa . A}}   \and
\inferrule*{\hastype{\Gamma, x : \later\kappa A}t{A}}{\hastype{\Gamma}{\fix\kappa x . t}{A}} 
 \and
\inferrule*{\hastype{\Gamma, \kappa : \clocktype}t{A}}{\hastype{\Gamma}{\Lambda \kappa . t}{\forall \kappa . A}} \and
\inferrule*{\hastype{\Gamma}t{\forall \kappa . A} \\ \wfclock{\Gamma}{\kappa'}}{\hastype{\Gamma}{\alwaysapp t{\kappa'}}{A \subst \kappa{\kappa'}}}  \and
\inferrule*{
\hastype{\Gamma}{t}{A} \\ \judgeeq{\Gamma}{A}{B}}
{\hastype{\Gamma}{t}{B}}
\end{mathpar}
\paragraph{Equalities}
\begin{align*}
 (\lambda x . t) u & = t\subst xu &
 \lambda x . t x & = t & \text{(if $x\notin t$)} \\
 \proj i{\pair{t_1}{t_2}} & = t_i &
 \pair{\proj 1t}{\proj 2t} & = t \\
 (\Lambda \kappa . t) \kappa' & = t\subst \kappa{\kappa'} &
 \Lambda \kappa . \alwaysapp t\kappa & = t & \text{(if $\kappa \notin t$)} \\ 
 \prev{\kappa} . \left(\nxt{\kappa}t\right) & = \Lambda\kappa . t &
 \nxt{\kappa}\left(\alwaysapp{\left(\prev{\kappa} . t\right)}{\kappa}\right) & = t \\ 
 \fix\kappa x . t & = t\subst x{\nxt{\kappa}(\fix\kappa x . t)} 
\end{align*}
\paragraph{Clock irrelevance axiom}
\[
\inferrule*{\hastype{\Gamma}t{\forall\kappa . A} \\ \kappa\notin \fv(A) \\ \wfclock{\Gamma}{\kappa'} \\ \wfclock{\Gamma}{\kappa''}}
{\eqjudg{\Gamma}{\alwaysapp t{\kappa'}}{\alwaysapp t{\kappa''}}[A]}
\]
\begin{center}

\caption{Syntax of $\coregdtt$, a fragment of $\gdtt$.}
\label{fig:basic:syntax}
\end{center}
\end{figure}

The rules for context formation, type judgements and equalities can be found in Figure~\ref{fig:basic:syntax}.
Note that $\clocktype$ has a special status.
In particular, it is not a type.
Its status is similar to that of the interval type in cubical type theory~\cite{Cubical}.
Ignoring $\later{}$ and the clock irrelevance axiom, the type theory $\coregdtt$ is in fact just a fragment of a type theory 
with a base type $\clocktype$ in which types like $\depprod xA\clocktype$ or $\depsum \kappa\clocktype A$ are not allowed.
Under this view, the type $\forall \kappa.A$ can be thought of as a dependent product type $\depprod \kappa\clocktype A$, in fact its basic behaviour is exactly like a dependent product, as can be seen from the equality rules.
We make use of this view to establish soundness of the model given in Section~\ref{sec:basics}.
What distinguishes it from an ordinary dependent product is the clock irrelevance axiom stated at the bottom 
of Figure~\ref{fig:basic:syntax}. The set $\fv(A)$ is the set of free variables of $A$ defined in the usual way, 
and so the assumption $\kappa \notin\fv(A)$ implies that $\forall\kappa . A$ reduces to a simple function 
space $\clocktype \to A$. The axiom states that all maps of this type are constant.
In Section~\ref{sec:orthogonality} we explain how to model the type theory with this additional axiom.

In Figure~\ref{fig:basic:syntax} the equalities should be understood as equalities of terms in a context.
For brevity we have omitted the context in most statements except the clock irrelevance axiom, which, unlike the other rules, is type directed.

The term constructor $\prev{\kappa}$ is a restricted elimination form for $\later{\kappa}$, and binds $\kappa$. 
An unrestricted eliminator 
$\prev{\kappa}$ would be unsafe, because terms of the form $\fix{\kappa} x . \prev{\kappa} . x$ would inhabit any type. As the 
model presented in this paper shows, however, it is safe to eliminate a $\later\kappa$, as long as $\kappa$ does not appear
in the ordinary (non-clock) variables of the context. This is ensured in the rule for $\prev{\kappa}$ by requiring that 
$\kappa$ is at the end of the context. One might have expected a simpler rule of the form
\begin{equation*}
\inferrule*{\hastype{\Gamma, \kappa : \clocktype}t{\later\kappa A}}
{\hastype{\Gamma,\kappa : \clocktype}{\prev{\kappa} . t}{A}}  
\end{equation*}
but this rule is not closed under substitution of clock variables. This problem is solved by binding $\kappa$. 

\subsubsection*{Some example terms.}
We refer to \citeasnoun{Bizjak-et-al:GDTT} for more extensive and detailed motivation and explanation of the usage of the type theory.
We briefly show here some example terms on streams.
The type $\gstream{\kappa}$ of \emph{guarded streams} of natural numbers is the unique type satisfying $\gstream{\kappa} \jeq \Nat \times \later{\kappa}{\gstream{\kappa}}$.
To understand this example it is not important how this type can be defined, only that it satisfies the stated judgemental equality.
For readers familiar with guarded dependent type theory we remark that it can be defined as using the guarded fixed point on the universe $\univ{\{\kappa\}}$ as outlined in the introduction of this paper.
Using the mentioned judgemental equality we can type
\begin{align*}
  \begin{split}
    \ghead{\kappa} &: \gstream{\kappa} \to \Nat\\
    \ghead{\kappa} &\defeq \lambda xs . \proj 1 (xs)
  \end{split}
  \begin{split}
    \gtail{\kappa} &: \gstream{\kappa} \to \later{\kappa}\gstream{\kappa}\\
    \gtail{\kappa} &\defeq \lambda xs . \proj 2 (xs)
  \end{split}
\end{align*}
Notice that the $\gtail{\kappa}$ introduces a $\later{\kappa}$ modality:
The tail of a guarded stream is only available later.
This prevents non-productive stream definitions.
However once the streams are defined we wish to use them without introducing later modalities.
This can be achieved by the type $\stream{}$ of \emph{streams} of natural numbers.
It is defined from the type of guarded streams as $\stream{} \defeq \forall\kappa.\gstream{\kappa}$.
Thus, the tail function on \emph{streams} is defined as
\begin{align*}
  \tail &: \stream{} \to \stream{}\\
  \tail &\defeq \lambda xs . \prev{\kappa} . \gtail{\kappa} (xs[\kappa])
\end{align*}

\subsection{Relation to previous presentations}
\label{sec:variation:syntax}

Judgements of $\gdtt$ as presented in \cite{Bizjak-et-al:GDTT} have a separate context for clock variables.
For example, typing judgements have the form $\Gamma \vdash_\Delta t : A$ where $\Delta$ is a clock context of the form 
$\kappa_1, \dots \kappa_n$, and $\Gamma$ consists exclusively of ordinary variable declarations.
The two presentations are equivalent in the sense that $\Gamma \vdash_\Delta t : A$ is a valid judgement in the presentation 
of \cite{Bizjak-et-al:GDTT} iff $\hastype{\kappa_1 : \clocktype, \dots, \kappa_n : \clocktype, \Gamma}{t}A$ is valid in the  
presentation used here.

Another minor difference is that $\gdtt$ as presented in \cite{Bizjak-et-al:GDTT} has a clock constant $\kappa_0$. The 
clock constant can be easily added to $\coregdtt$ by a precompilation adding a fresh clock variable to the left of the context
in each judgement. 

\section{A presheaf model}
\label{sec:basics}

This section defines the category $\BaseCat$ as that of covariant presheaves on the category of \emph{time objects} 
$\TimeCat$. 
As any presheaf category, $\BaseCat$ has enough structure to model 
dependent type theory.
The category $\BaseCat$ contains an object $\Cl$ of \emph{clocks} which can be used to model clock 
quantification and guarded recursion. We show that $\BaseCat$ validates almost all the rules of $\coregdtt$, 
apart from the clock irrelevance axiom, which is the topic of Section~\ref{sec:orthogonality}. The focus in this section, as
in most of the paper, will be to construct the semantic structure needed for modelling the type theory, leaving the question
of how to interpret syntax to Section~\ref{sec:interp:syntax}.

We write $\Fin$ for the category of finite sets and functions whose objects are finite subsets $\FSA$ of some given, countably infinite, set of clocks.%
\footnote{The assumption that the objects are subset of a fixed set, as opposed to arbitrary finite sets keeps the category $\Fin$, and thus also $\TimeCat$, small, thus simplifying definitions of, e.g., dependent products.}
\begin{definition}
  \label{def:time-cat-defn}
  Let $\TimeCat$ be the category with objects pairs $(\FSA, \delta)$ with $\FSA \in \Fin$ and $\delta : \FSA \to \NN$ a function.
  A morphism $(\FSA, \delta) \to (\FSB, \delta')$ in $\TimeCat$ is a \emph{function} $\tau : \FSA \to \FSB$ in $\Fin$ such that $\delta' \comp \tau \leq \delta$ in the pointwise ordering.
\end{definition}

We use $\lambda$ to range over elements of $\FSA$ and write 
$\FSA, \lambda$ for the union of $\FSA$ with $\{\lambda\}$ assuming $\lambda \notin \FSA$. Likewise, when $\FSA$ and
$\FSB$ are disjoint, we write $\FSA,\FSB$ for their union. We use the notation 
$\delta[\lambda\mapsto n]$ for both the update of $\delta$ (when $\lambda \in \FSA$) and the extension of $\delta$ 
(when $\lambda \notin \FSA$). 

The indexing category $\TimeCat$ should be thought of as a category of \emph{time objects}.
A \emph{time object} is a finite set of semantic clocks $\FSA$ which each have a finite number of 
ticks left on them as indicated by $\delta$.
During a computation, three things can happen: time can pass on the existing clocks, as captured by a map 
$\id\FSA : (\FSA,\delta) \to (\FSA, \delta')$ where $\delta'\leq \delta$, new clocks can be introduced as 
captured by set inclusions $i : (\FSA,\delta) \to ((\FSA,\lambda), \delta[\lambda\mapsto n])$, and clocks can 
be synchronised as captured by a map
\begin{align*}
  \id\FSA[\lambda \mapsto \lambda'', \lambda'\mapsto \lambda''] : ((\FSA, \lambda, \lambda'),\delta[\lambda\mapsto n, \lambda'\mapsto m]) \to ((\FSA, \lambda''),\delta[\lambda''\mapsto \min(n, m)]).
\end{align*}
Finally, clocks can be renamed, e.g., via an isomorphism $\sigma : \FSA \iso \FSB$ inducing an isomorphism
$\sigma : (\FSA, \delta\circ \sigma) \to (\FSB, \delta)$. 
Any map in the indexing category $\TimeCat$ can be written as a composition of these three kinds of maps.

Variables of the form $\kappa : \clocktype$ will be modelled as the \emph{object of clocks} $\Cl$, 
which is simply the first projection
\begin{align*}
  \Cl(\FSA,\delta) = \FSA.
\end{align*}

\begin{lemma}
  \label{lem:object-of-clocks-is-colimit}
  Let $\lambda$ be a clock.
  There is an isomorphism of objects of $\BaseCat$
  \begin{align*}
    \Cl \cong \colim_{n \in \NN} y\left(\{\lambda\},n\right)
  \end{align*}
  where $y : \TimeCat^{\text{op}} \to \BaseCat$ is the (co)Yoneda embedding, and we write
  $\left(\{\lambda\},n\right)$ for the $\TimeCat$ object $\left(\{\lambda\}, [\lambda \mapsto n]\right)$, i.e., the first component is the singleton containing $\lambda$, and the second component is the map which maps $\lambda$ to $n$.
\end{lemma}
\begin{proof}
 The objects of the diagram are 
\begin{align*}
 y\left(\{\lambda\},n\right)(\FSA,\delta) & = \Hom_{\TimeCat}(\left(\{\lambda\},n\right), (\FSA,\delta)) \\
 & \iso \{\lambda'\in \FSA \mid \delta(\lambda') \leq n \}
\end{align*}
and up to this isomorphism, the arrows are inclusions of sets. Since colimits are computed pointwise in presheaf 
categories, the isomorphism follows. 
\end{proof}

When describing objects and morphisms of $\BaseCat$ we will use the following notation: An 
object $\Gamma$ is a family of sets $\Gamma_{(\FSA,\delta)}$ indexed 
by $(\FSA,\delta) \in \TimeCat$ together with maps 
\begin{align*}
  \sigma \cdot - : \Gamma_{(\FSA,\delta)} \to \Gamma_{(\FSB,\delta')}
\end{align*}
for each $\sigma : (\FSA,\delta) \to (\FSB,\delta')$ in $\TimeCat$, satisfying the following two functoriality properties
\begin{align}
  \id{} \cdot x &= x \label{eq:id:pres} \\
  \left(\sigma \comp \tau\right) \cdot x &= \sigma \cdot (\tau \cdot x). \label{eq:comp:pres}
\end{align}
A morphism $\rho : \Gamma \to \Gamma'$ is a family of maps $\rho_{(\FSA,\delta)} : \Gamma_{(\FSA,\delta)}
\to \Gamma'_{(\FSA,\delta)}$ such that $\sigma \cdot(\rho_{(\FSA,\delta)}(\gamma)) = \rho_{(\FSB,\delta')}(\sigma\cdot \gamma)$ 
for any $\sigma : (\FSA,\delta) \to (\FSB,\delta')$ in $\TimeCat$ and any
$\gamma \in \Gamma_{(\FSA,\delta)}$. 


\subsection{Interpreting type theory in categories of presheaves}
\label{sec:interpreting-tt-in-presheaves}

We now recall the notion of category with families (CwF)~\cite{dybjer1996}, 
which is a standard notion of model of dependent type theory.
We also recall how $\BaseCat$ gives rise to a CwF modelling $\Pi$-, and $\Sigma$-types through a 
standard construction~\cite{Hofmann:syntax-and-semantics} that works for all presheaf categories. 

\begin{definition}
 A category with families comprises
 \begin{itemize}
\item A category $\cat C$ with a distinguished terminal object
\item For each object $\Gamma$ of $\cat C$ a set $\Fam[\cat C]{\Gamma}$ of \emph{families} over $\Gamma$. 
\item For each $\Gamma$ in $\cat C$ and each family $A$ in $\Fam[\cat C]{\Gamma}$ a set $\ElCwF[\cat C]{\Gamma}{A}$
of \emph{elements} of $A$. 
\item For each morphism $\gamma : \Delta \to \Gamma$ in $\cat C$ \emph{reindexing} operations
mapping $A$ in $\Fam[\cat C]\Gamma$ to $A[\gamma]$ in $\Fam[\cat C]\Delta$ and $t$ in $\ElCwF[\cat C]{\Gamma}A$ to 
$t[\gamma]$ in $\ElCwF[\cat C]\Delta{A[\gamma]}$. 
These must
satisfy the equations $A[\id{}] = A$, $t[\id{}] = t$,  
$A[\gamma\circ \delta] = A[\gamma][\delta]$ and $t[\gamma\circ \delta] = t[\gamma][\delta]$
for all morphisms $\delta$ with codomain $\Delta$.  
\item A \emph{comprehension} operation associating to each family $A$ in $\Fam[\cat C]{\Gamma}$ 
the following: An object $\compr \Gamma A$ in $\cat C$, a morphism 
$\p_A : \compr \Gamma A \to \Gamma$ and an element ${\q_A}$ in $\ElCwF[\cat C]{\compr\Gamma A}{A[\p_A]}$, 
such that for every
$\gamma : \Delta \to \Gamma$, and $t$ in $\ElCwF[\cat C]\Delta{A[\gamma]}$ there exists a unique morphism 
$\cpair \gamma t : \Delta \to \compr \Gamma A$ such that  $\p_A \circ\cpair \gamma t = \gamma$ and 
$\q_A[\cpair \gamma t] = t$. 
\end{itemize}
\end{definition}

Note that uniqueness implies that $\cpair\gamma t \circ \rho = \cpair{\gamma \circ \rho}{t[\rho]}$.

We will often refer to a CwF simply by its underlying category $\cat C$ leaving the rest of the structure
implicit. Categories with families provide models of dependent type theories in which contexts are 
interpreted as objects in the underlying category, types are interpreted as families and terms
as elements. The category $\BaseCat$ is the underlying category of a CwF whose families above an
object $\Gamma$ are families of sets $A_{(\FSA,\delta)}(\gamma)$ indexed over $(\FSA,\delta)$ 
in $\TimeCat$ and $\gamma \in \Gamma_{(\FSA,\delta)}$, together with restriction maps 
\begin{align*}
  \sigma \cdot (-) : A_{(\FSA,\delta)}(\gamma) \to A_{(\FSB,\delta')}(\sigma \cdot \gamma)
\end{align*}
indexed by $\sigma : (\FSA,\delta) \to (\FSB,\delta')$ in $\TimeCat$ and $\gamma \in \Gamma_{(\FSA,\delta)}$ and 
satisfying the functoriality properties (\ref{eq:id:pres}) and (\ref{eq:comp:pres}).
Note that the notation $\sigma \cdot x$ is overloaded both for a restriction of objects as well as families.

An element $t \in \ElCwF{\Gamma}A$ is a family of elements $t_{(\FSA,\delta)}(\gamma) \in A_{(\FSA,\delta)}(\gamma)$ 
indexed over $(\FSA,\delta)$ in $\TimeCat$ and $\gamma \in \Gamma_{(\FSA,\delta)}$ satisfying 
$\sigma\cdot (t_{(\FSA,\delta)}(\gamma)) = t_{(\FSB,\delta')}(\sigma \cdot \gamma)$ for every 
$\sigma : (\FSA,\delta) \to (\FSB,\delta')$. Reindexing of terms and types along morphisms $\rho : \Gamma' \to \Gamma$
is defined as ${A[\rho]}_{(\FSA,\delta)}(\gamma) = A_{(\FSA,\delta)}(\rho_{(\FSA,\delta)}(\gamma))$ 
and  ${t[\rho]}_{(\FSA,\delta)}(\gamma) = t_{(\FSA,\delta)}(\rho_{(\FSA,\delta)}(\gamma))$. We will often omit the subscripts
$(\FSA,\delta)$ when they can be inferred from the context. 

Comprehension is defined as 
\[(\compr\Gamma A)_{(\FSA,\delta)} = \{(\gamma, a) \mid \gamma \in \Gamma_{(\FSA,\delta)}, a \in A_{(\FSA,\delta)}(\gamma)\}
\]
with presheaf action defined as $\sigma \cdot (\gamma,x) = (\sigma \cdot \gamma, \sigma \cdot x)$.

Recall the following standard result~\cite{Hofmann:syntax-and-semantics}.

\begin{lemma}
 The CwF structure on $\BaseCat$ models $\Pi$- and $\Sigma$-types. 
\end{lemma}

These are constructed as follows, for $A \in \Fam\Gamma$, $B \in \Fam{\compr\Gamma A}$ and 
$\gamma\in \Gamma_{(\FSA,\delta)}$
\begin{align}
 \PiSem AB_{(\FSA,\delta)}(\gamma) & =     \left\{ \left(f_{\sigma}\right)_{\sigma : (\FSA,\delta) \to (\FSB, \delta')} \isetsep
      \begin{array}{l}
        \forall \FSB, \delta', \sigma : (\FSA,\delta) \to (\FSB,\delta'), \forall a \in A_{(\FSB,\delta')}(\sigma \cdot \gamma),\\
        f_\sigma(a) \in B_{(\FSB,\delta')}(\sigma \cdot \gamma, a) \text{ such that}\\
        \tau \cdot f_{\sigma}(a) = f_{\tau \comp \sigma}(\tau \cdot a) \text{ for composable } \tau, \sigma
      \end{array}
    \right\}  \label{eq:Pi:Sem} \\
  \SigmaSem AB_{(\FSA,\delta)}(\gamma) & =  \{ (a,b) \mid a \in A_{(\FSA,\delta)}(\gamma), b \in B_{(\FSA,\delta)}(\gamma, a)\} 
   \label{eq:Sigma:Sem} 
\end{align}
with presheaf action on $\PiSem AB$ defined by precomposition, i.e., if $\tau : (\FSA,\delta) \to (\FSB,\delta')$ then
\begin{align*}
  \tau \cdot (\left(f_{\sigma}\right)_{\sigma : (\FSA,\delta) \to (\FSC,\delta'')}) =
  \left(f_{\sigma \tau}\right)_{\sigma : (\FSB,\delta') \to (\FSC,\delta'')}
\end{align*}
Recall also that evaluation mapping an element $f \in \ElCwF\Gamma{\PiSem AB}$ and $t \in \ElCwF\Gamma A$ 
to $\ev(f,t) \in \ElCwF\Gamma{B[\cpair{\id\Gamma}{t}]}$ is defined as 
\[
\ev(f,t)_{(\FSA,\delta)}(\gamma) = (f_{(\FSA,\delta)}(\gamma))_{\id{(\FSA, \delta)}}(t_{(\FSA,\delta)}(\gamma))
\]
When $A,B \in \Fam\Gamma$ we write $A \to B$ for $\PiSem{A}{B[\p]}$.  When $t \in \ElCwF{\compr\Gamma A}B$
we write $\lambda(t)$ for the corresponding abstracted element in $\ElCwF\Gamma{\PiSem AB}$. The semantic
$\beta$-rule states that $\ev(\lambda t, u) = t[\cpair{\id{}}u]$. Finally, recall the substitution property 
$\PiSem AB[\rho] = \PiSem{A[\rho]}{B[\cpair{\p \rho}{\q}]}$, and similarly for $\Sigma$-types. 

\subsection{Modelling \texorpdfstring{$\later{}$}{later} and guarded recursion}
\label{sec:modelling-later-and-guarded-rec}

We now explain how to model the $\later{}$-modality and fixed points. 
First note that there is a family $\clockfam \in \Fam 1$ defined as the object $\Cl$, since families in context $1$ correspond to
objects of $\BaseCat$, and so, for any $\Gamma$, there is a family $\clockfam[\uniquemap_\Gamma]\in \Fam\Gamma$, 
where $\uniquemap_\Gamma : \Gamma \to 1$ is the unique map. 

\begin{lemma} \label{lem:latercwf}
 If $A\in \Fam\Gamma$ and $\kappa \in \ElCwF\Gamma{\clockfam[\uniquemap_\Gamma]}$ 
 there is a family $\laterfam\kappa A \in \Fam\Gamma$ and a mapping associating to each element
 $t \in \ElCwF\Gamma A$ an element $\nextel\kappa (t) \in \ElCwF\Gamma{\laterfam\kappa A}$ both commuting
 with reindexing, such that for every $f \in \ElCwF{\Gamma}{\laterfam\kappa A\to A}$ there is a unique
 $\fixel\kappa f \in \ElCwF{\Gamma}A$ satisfying $\ev(f, \nextel\kappa(\fixel\kappa f)) = \fixel\kappa f$. 
\end{lemma}
Note that the uniqueness here implies that the construction $\fixel\kappa f$ commutes with reindexing: Since
\begin{align*}
 \ev(f[\rho], \nextel{\kappa[\rho]}((\fixel\kappa f)[\rho])) & = \ev(f[\rho], (\nextel\kappa(\fixel\kappa f))[\rho]) \\
 & =  \ev(f, \nextel\kappa(\fixel\kappa f))[\rho] \\
 & =  (\fixel\kappa f)[\rho] 
\end{align*}
uniqueness implies $(\fixel\kappa f)[\rho] = \fixel{\kappa[\rho]}{f[\rho]}$. 

\begin{proof}
If $(\FSA,\delta)$ is an object of $\TimeCat$ and $\lambda \in \FSA$ such that $\delta(\lambda) > 0$ 
we write $\minkappa{\delta}{\lambda}$ for the function which agrees with $\delta$ 
everywhere except on $\lambda$ where 
$\minkappa{\delta}{\lambda}(\lambda) = \delta(\lambda) - 1$.
It is elementary that the identity function defines a morphism
\begin{align*}
  \tick{\lambda}{\FSA} : (\FSA,\delta) \to \left(\FSA,\minkappa{\delta}{\lambda}\right)
\end{align*}
in $\TimeCat$.

With this notation we can define $\laterfam\kappa A$ as follows, omitting the subscript on $\kappa$
\begin{align*}
  (\laterfam\kappa A)_{(\FSA,\delta)}(\gamma) =
  \begin{cases}
    \{\star\} & \text{ if } \delta(\kappa(\gamma)) = 0\\
    A_{(\FSA,\minkappa{\delta}{\kappa(\gamma)})}(\tick{\kappa(\gamma)}{\FSA} \cdot \gamma) & \text{ otherwise }
  \end{cases}
\end{align*}
Let $\sigma : (\FSA,\delta) \to (\FSB,\delta')$. The map $\sigma \cdot (-) : (\laterfam\kappa A)_{(\FSA,\delta)}(\gamma)
\to (\laterfam\kappa A)_{(\FSB,\delta')}(\sigma\cdot \gamma)$, can be defined in the case that 
$\delta'(\kappa(\sigma\cdot \gamma)) = 0$ as $\sigma\cdot x = \star$. If $\delta'(\kappa(\sigma\cdot \gamma)) > 0$ 
also $\delta(\kappa(\gamma)) > 0$ because
\[
\delta'(\kappa(\sigma\cdot \gamma)) = \delta'(\sigma\cdot \kappa(\gamma)) = \delta'(\sigma(\kappa(\gamma))) 
\leq \delta(\kappa(\gamma)) 
\]
and so $\sigma$ induces a map $\minkappa\sigma{\kappa(\gamma)} : (\FSA,\minkappa{\delta}{\kappa(\gamma)})
\to (\FSB,\minkappa{\delta'}{\sigma(\kappa(\gamma))})$, satisfying 
$\minkappa\sigma{\kappa(\gamma)}\circ \tick{\kappa(\gamma)}{\FSA} = \tick{\sigma(\kappa(\gamma))}{\FSA} \circ \sigma$.
In this case, we can thus define $\sigma\cdot(-)$ to be the map
\[
\minkappa\sigma{\kappa(\gamma)} \cdot (-) : 
A_{(\FSA,\minkappa{\delta}{\kappa(\gamma)})}(\tick{\kappa(\gamma)}{\FSA} \cdot \gamma)
\to A_{(\FSA,\minkappa{\delta'}{\kappa(\sigma\cdot \gamma)})}(\tick{\kappa(\sigma\cdot \gamma)}{\FSB} \cdot \sigma\cdot\gamma)
\]
 
The construction $\laterfam\kappa$ commutes with reindexing, since
\begin{align*}
  (\laterfam{\kappa[\rho]} A[\rho]) & =
  \begin{cases}
    \{\star\} & \text{ if } \delta(\kappa[\rho](\gamma)) = 0\\
    A_{(\FSA,\minkappa{\delta}{\kappa[\rho](\gamma)})}(\rho(\tick{\kappa[\rho](\gamma)}{\FSA} \cdot \gamma)) & \text{ otherwise }
  \end{cases} \\
  & =
  \begin{cases}
    \{\star\} & \text{ if } \delta(\kappa(\rho(\gamma))) = 0\\
    A_{(\FSA,\minkappa{\delta}{\kappa(\rho(\gamma))})}(\tick{\kappa(\rho(\gamma))}{\FSA} \cdot \rho(\gamma)) & \text{ otherwise }
  \end{cases}
\end{align*}
and writing out $((\laterfam\kappa A)[\rho])(\gamma) = (\laterfam\kappa A)(\rho(\gamma))$ gives the exact same expression. 

Analogously, the element $\nextel{\kappa}(t)$ is defined as
\begin{align*}
  (\nextel{\kappa}{(t)})_{(\FSA,\delta)}(\gamma) =
  \begin{cases}
    \star & \text{ if }  \delta(\kappa(\gamma)) = 0\\
    t_{(\FSA,\minkappa{\delta}{\kappa(\gamma)})}(\tick{\kappa(\gamma)}{\FSA} \cdot \gamma) & \text{ otherwise }
  \end{cases}
\end{align*}
To define $\fixel\kappa f$, note that by the above definitions
\begin{align*}
 \ev(f, \nextel\kappa(\fixel\kappa f))(\gamma)
   & =
  \begin{cases}
    (f(\rho))_{\id{}}(\star) & \text{ if } \delta(\kappa(\rho(\gamma))) = 0\\
    (f(\rho))_{\id{}}((\fixel\kappa f)(\tick{\kappa(\gamma)}{\FSA} \cdot \gamma)) & \text{ otherwise }
  \end{cases}
\end{align*}
and thus the $\fixel\kappa f$ can be defined by induction on $\kappa(\rho(\gamma))$. 
\end{proof}

\subsection{Modelling previous}
\label{sec:prev}

As noted in Section~\ref{sec:basic:syntax}, universal quantification over clocks is simply a special 
case of a dependent function space, and so can be modelled in the CwF $\BaseCat$ using $\Pi$-types. 
No special construction is needed for this. However, in order to model $\prev{}$ we now give an 
alternative description of $\Pi$ types with domain $\clockfam$ in the model as a limit over a 
family of objects indexed by natural numbers. Universal quantification over clocks is modelled 
similarly in the models of~\cite{Atkey:Productive,Mogelberg:tt-productive-coprogramming,Bizjak-Moegelberg:clocks-model}.

In the following, we will assume a choice of fresh clock names $\FSA \mapsto \fresh{\FSA}$, such that 
$\fresh{\FSA} \notin \FSA$ and write 
$\iota^n : (\FSA,\delta) \to \left((\FSA,\fresh\FSA), \delta[\fresh\FSA\mapsto n]\right)$ for the inclusion for $n \in \NN$. Note that 
$\tick\lambda{(\FSA,\fresh\FSA)} \circ \iota^{n+1} = \iota^n$.

\begin{lemma} \label{lem:forall:interp}
 Let $A \in \Fam{\compr\Gamma{\clockfam[\uniquemap_\Gamma]}}$, 
 $(\FSA,\delta) \in \TimeCat$ and $\gamma \in \Gamma_{(\FSA,\delta)}$. 
 The set
   $\PiSem {\clockfam[\uniquemap_\Gamma]} A_{(\FSA,\delta)}(\gamma)$
 is the limit of the diagram
 \begin{displaymath}
   \begin{diagram}
     A(\iota^0\cdot \gamma, \fresh\FSA) & \ar{l}[swap]{\tick{\fresh\FSA}{} \cdot (-)} A(\iota^1\cdot \gamma, \fresh\FSA) & \ar{l}[swap]{\tick{\fresh\FSA}{} \cdot (-)} A(\iota^2\cdot \gamma, \fresh\FSA) & \ar{l}[swap]{\tick{\fresh\FSA}{} \cdot (-)} \dots
   \end{diagram}
 \end{displaymath}
\end{lemma}

\begin{proof}
 Let $y : \TimeCat^{\text{op}} \to \BaseCat$ be the (co)Yoneda embedding. Uncurrying the definition in (\ref{eq:Pi:Sem})
 we see that elements of 
 $\PiSem{\clockfam[\uniquemap_\Gamma]} A_{(\FSA,\delta)}(\gamma)$ correspond to maps mapping objects 
 $(\FSB,\delta')$ of $\TimeCat$ and elements $(\sigma,\lambda) \in (y(\FSA,\delta)\times \Cl)(\FSB,\delta')$ 
 to elements in $A(\sigma\cdot \gamma, \lambda)$ naturally in $(\FSB,\delta')$. By 
 Lemma~\ref{lem:object-of-clocks-is-colimit}, the presheaf $y(\FSA,\delta)\times \Cl$
 is isomorphic to the colimit over $n$ of the diagram given by objects $y(\FSA,\delta)\times y(\{\fresh\FSA\}, n)$.
 Thus $\PiSem{\clockfam[\uniquemap_\Gamma]} A_{(\FSA,\delta)}(\gamma)$
 is isomorphic to the limit of a diagram of the form
 \begin{displaymath}
   \begin{diagram}
     X_0 &  \ar{l} X_1 &  \ar{l} X_2 &  \ar{l} X_3 &  \ar{l} \dots
   \end{diagram}
 \end{displaymath}
 where $X_n$ is the set of maps as above defined just for 
 $(\sigma,\lambda) \in (y(\FSA,\delta)\times y(\{\fresh\FSA\}, n))(\FSB,\delta')$, and the maps are given by restriction.
 It remains to show the isomorphism of the above diagram with that of the lemma.
 
 The object $((\FSA, \fresh\FSA), \delta[\fresh\FSA\mapsto n])$ is the coproduct in $\TimeCat$ of $(\FSA,\delta)$
 and $(\{\fresh\FSA\}, n)$ with inclusions given by inclusions of sets. Since the yoneda embedding preserves
 products, ${y(\FSA,\delta)\times y(\{\fresh\FSA\}, n)} \iso y((\FSA, \fresh\FSA), \delta[\fresh\FSA\mapsto n])$. Up 
 to this correspondence, the restriction of an element in the family
 $\PiSem{\clockfam[\uniquemap_\Gamma]} A_{(\FSA,\delta)}(\gamma)$
 to $y(\FSA,\delta)\times y(\{\fresh\FSA\}, n)$ corresponds to a mapping associating to 
 $(\FSB,\delta')$ of $\TimeCat$ and elements 
 $(\sigma,\lambda) \in y((\FSA, \fresh\FSA), \delta[\fresh\FSA\mapsto n])(\FSB,\delta')$ 
 elements in $A(\sigma\cdot (\iota^n\cdot \gamma), \sigma(\fresh\FSA))$ naturally in $(\FSB,\delta')$.
 By a yoneda style argument such mappings are determined by their action on the identity on 
 $((\FSA, \fresh\FSA), \delta[\fresh\FSA\mapsto n])$ and thus we arrive at the diagram of the lemma. 
 \end{proof}

Rather than modeling $\prev{}$ directly, we model the construct 
\[
\inferrule*
{\hastype{\Gamma}{t}{\forall\kappa . \later\kappa A}}
{\hastype{\Gamma}{\force{t}}{\forall\kappa . A}}
\]
Using this, one can define $\prev{}$ as
\begin{align}
 \prev{\kappa} . t &\eqdef \force(\Lambda \kappa . t) \label{def:prev}
\end{align}
To satisfy the equalities of Figure~\ref{fig:basic:syntax}, the term $\lambda x .\force(x)$ should be an inverse
to $\lambda x . \Lambda\kappa . \nxt\kappa (\alwaysapp{x}\kappa)$. Using this, one can prove the first equality
for $\force$ in Figure~\ref{fig:basic:syntax} as follows
\begin{align*}
\prev{\kappa} . \left(\nxt{\kappa}t\right) & = \force (\Lambda\kappa . \left(\nxt{\kappa}t\right)) \\
& = \force (\Lambda\kappa . \left(\nxt{\kappa}\alwaysapp{(\Lambda\kappa .t)}{\kappa}\right)) \\
& = \Lambda\kappa . t
\end{align*}
The other equality is proved similarly.

We now show that the semantic correspondent to $\lambda x . \Lambda\kappa . \nxt\kappa (\alwaysapp{x}\kappa)$
is an isomorphism. 

\begin{lemma}
 Suppose $A\in \Fam{\compr{\Gamma}{\clockfam[\uniquemap_\Gamma]}}$.
 The mapping of elements $t \in \ElCwF{\Gamma}{\PiSem {\clockfam[\uniquemap_\Gamma]}A}$
 to $\lambdael{\nextel{\q}(\ev(t[\p],\q))}$ in 
 $\ElCwF{\Gamma}{\PiSem {\clockfam[\uniquemap_\Gamma]}{\laterfam{\q}A}}$ is an isomorphism.
\end{lemma}

Before proving this, we argue that the mapping referred to is welltyped. By the assumption on $t$,
$t[\p] \in \ElCwF{\compr{\Gamma}{\clockfam[\uniquemap_\Gamma]}}{(\PiSem {\clockfam[\uniquemap_\Gamma]}A)[\p]}$ 
and since $(\PiSem {\clockfam[\uniquemap_\Gamma]}A)[\p]$ equals 
$\PiSem{\clockfam[\uniquemap_{\compr{\Gamma}{\clockfam[\uniquemap_\Gamma]}}]}{A[\cpair{\p\circ\p}{\q}]}$ 
also $\ev(t[\p],\q)$ is an element in $A[\cpair{\p\circ\p}{\q}][\cpair{\id{}}{\q}] = A[\cpair{\p}{\q}] = A$. Since
$\q \in \ElCwF{\compr{\Gamma}{{\clockfam[\uniquemap_\Gamma]}}}$ also $\nextel{\q}(\ev(t[\p],\q))$
is an element in $\laterfam{\q}A$, and therefore $\lambdael{\nextel{\q}(\ev(t[\p],\q))}$ is in 
$\ElCwF{\Gamma}{\PiSem {\clockfam[\uniquemap_\Gamma]}{\laterfam{\q}A}}$. Note that the inverse
of this map necessary must commute with reindexing, since the construction of the map does. 

\begin{proof}
 Unfolding definitions, we see that the construction of lemma at $\gamma \in \Gamma_{(\FSA, \delta)}$ is the 
 map induced by the map of diagrams below. 
 \begin{displaymath}
   \begin{diagram}
     A(\iota^0\cdot \gamma, \fresh\FSA) \ar{d}[swap]{\uniquemap} 
     & \ar{l}[swap]{\tick{\fresh\FSA}{} \cdot (-)} A(\iota^1\cdot \gamma, \fresh\FSA) \ar{d}{\tick{\fresh\FSA}{} \cdot (-)}
     & \ar{l}[swap]{\tick{\fresh\FSA}{} \cdot (-)} A(\iota^2\cdot \gamma, \fresh\FSA) \ar{d}{\tick{\fresh\FSA}{} \cdot (-)}
     & \ar{l}[swap]{\tick{\fresh\FSA}{} \cdot (-)} \dots \\  
     1 & \ar{l} A(\iota^0\cdot \gamma, \fresh\FSA) & \ar{l}[swap]{\tick{\fresh\FSA}{} \cdot (-)} A(\iota^1\cdot \gamma, \fresh\FSA) & \ar{l}[swap]{\tick{\fresh\FSA}{} \cdot (-)}  \dots 
   \end{diagram}
 \end{displaymath}
 The map induced between the limits is therefore an isomorphism. 
\end{proof}

With these definitions we can extend the interpretation to the whole of $\coregdtt$.
However the interpretation only validates the basic axioms, i.e., $\beta$ and $\eta$ laws.
It does not validate the clock irrelevance axiom.
To soundly interpret $\coregdtt$ we need to require that the families are suitably constant.
This is the subject of the next section. 

\section{Modelling clock irrelevance using orthogonality}
\label{sec:orthogonality}

In the interpretation above $\forall\kappa.A$ is interpreted as an ordinary dependent product $\depprod{\kappa}{\clocktype}A$.
Under this interpretation, the clock irrelevance axiom concerns functions $f$ of type $\clocks \to A$ and states
that each such function must be constant.
To model this, we restrict attention in the model to those families satisfying this property, and show that the collection of
these is closed under the type constructions of $\coregdtt$.
To capture clock irrelevance semantically, we start by recalling the category theoretic concept of orthogonality. 

A morphism $e : A\to B$ is \emph{left-orthogonal} to $m : C\to D$
(and $m$ is right-orthogonal to $e$) if all commutative squares as below have
a \emph{unique} filler $h$. 
\begin{displaymath}
  \begin{diagram}
    A\ar{r}{f} \ar{d}{e} & C\ar{d}{m}\\
    B \ar[dotted]{ur}[description]{h} \ar{r}{g} & D
  \end{diagram}
\end{displaymath}
Often we will simply refer to this as $e$ being orthogonal to $m$. If $B$ is the terminal object, we may 
also refer to this as the object $A$ being left-orthogonal to $m$ and similarly for the case of $D$ being terminal. 
We shall need the slightly stronger notion of \emph{internal orthogonality}~\cite{anel2017generalized}, 
which can be understood by rephrasing the above lifting property as the requirement that the following diagram of
hom-sets is a pullback
\begin{displaymath}
  \begin{diagram}
    \Hom_{\ccat}(B,C)\ar{r}{m\circ(-)} \ar{d}{(-)\circ e} & \Hom_{\ccat}(B,D)\ar{d}{(-)\circ e}\\
    \Hom_{\ccat}(A,C) \ar{r}{m\circ(-)} & \Hom_{\ccat}(A,D)
  \end{diagram}
\end{displaymath}
The idea of internal orthogonality is to replace the external hom-sets above with exponentials in a cartesian
closed category. The resulting condition is equivalent to the following, which can be stated also in categories
that are not cartesian closed.

\begin{definition}
  \label{def:invariant-under-clock-introduction}
  Let $\ccat$ be a category with finite products. Say a morphism $p : A \to B$ is \emph{internally right orthogonal to an object $X$} if
  for any $Y$ and any $f,g$ making the outer square below commute,
  there exists a \emph{unique} $h : Y \to A$ such that the diagram
  \begin{displaymath}
    \begin{diagram}
      Y \times X \ar{r}{f} \ar{d}{\pi_Y} & A\ar{d}{p}\\
      Y \ar{r}{g} \ar[dotted]{ur}[description]{h} & B
    \end{diagram}
  \end{displaymath}
  commutes.
  
  A map $p : A \to B$ in $\BaseCat$ is \emph{invariant under clock introduction} if it is internally right orthogonal to any object of the form $y\left(\{\lambda\},n\right)$. 
\end{definition}

\begin{definition}
  \label{def:types-invariant-under-clock-intro}
A family $A \in \Fam\Gamma$ is \emph{invariant under clock introduction} if  
$\p : \compr\Gamma A \to \Gamma$ is invariant under clock introduction in the sense of Definition~\ref{def:invariant-under-clock-introduction}.
\end{definition}

The terminology of being invariant under clock introduction is justified by the following lemma, the proof of which is on page~\pageref{proof:lem:fibrations-iff-restrictions-pullback} after preliminary Lemmas~\ref{lem:characterization-of-fibrations} and~\ref{lem:A:to:yoneda}.

\begin{lemma}
  \label{lem:fibrations-iff-restrictions-pullback}
  A morphism $p : A \to B$ in $\BaseCat$ is invariant under clock introduction if and only if for all $(\FSA,\delta) \in \TimeCat$, and any (equivalently all) $\lambda \not\in\FSA$ and any $n$ the square
  \begin{displaymath}
    \begin{diagram}
      \pullbacktip A(\FSA,\delta) \ar{d}{p_{(\FSA,\delta)}}\ar{r}{A(\iota^n)} & A\left((\FSA,\lambda), \delta[\lambda\mapsto n]\right) \ar{d}{p_{\left((\FSA,\lambda), \delta[\lambda\mapsto n]\right)}}\\
      B(\FSA,\delta) \ar{r}{B(\iota^n)} & B\left((\FSA,\lambda), \delta[\lambda\mapsto n]\right)
    \end{diagram}
  \end{displaymath}
  is a pullback, where 
  $\iota^n : (\FSA, \delta) \to \left((\FSA,\lambda), \delta[\lambda\mapsto n]\right)$ is the inclusion.
\end{lemma}

In particular, for any presheaf $A$, the unique map $A \to 1$ is invariant under clock introduction iff $A$ is a constant presheaf.
It will be an invariant of the interpretation defined here that the interpretation of any type is
invariant under clock introduction. 

  Lemma~\ref{lem:fibrations-iff-restrictions-pullback} can be restated in the following way for interpretations of types.

\begin{corollary}
  \label{cor:invariance-under-clock-intro-for-types}
  A family $A \in \Fam\Gamma$ is invariant under clock introduction if and only if for any $\FSA$, any $\lambda \not \in \FSA$, any inclusion $\iota^n : (\FSA, \delta) \to \left((\FSA,\lambda), \delta[\lambda\mapsto n]\right)$, and any $\gamma \in \Gamma_{(\FSA,\delta)}$, the action 
  \begin{equation}
    \label{eq:invariance:under:clock:intro}
    \iota^n\cdot (-) : A_{(\FSA,\delta)}(\gamma) \to
    A_{\left((\FSA,\lambda), \delta[\lambda\mapsto n]\right)}(\iota^n\cdot\gamma)
  \end{equation}
  is an isomorphism.
\end{corollary}
The proof of Lemma~\ref{lem:fibrations-iff-restrictions-pullback} uses the characterisation of internal orthogonality in Lemma~\ref{lem:characterization-of-fibrations} together with the characterisation of exponentiation with certain representable functors in Lemma~\ref{lem:A:to:yoneda}.

The following lemma is proved by a straightforward diagram chase.
\begin{lemma}
  \label{lem:characterization-of-fibrations}
  Suppose $\CC$ is cartesian closed, $X$ is an object of $\CC$ and $p : A \to B$ a morphism.
  Then $p$ is internally right orthogonal to $X$ if and only if
  \begin{displaymath}
    \begin{diagram}
      A \pullbacktip \ar{d}{p} \ar{r}{c_A} & A^X \ar{d}{p^X}\\
      B \ar{r}{c_B} & B^X
    \end{diagram}
  \end{displaymath}
  is a pullback.
  Here $c_A$ and $c_B$ are exponential transposes of projections $A \times X \to A$ and $B \times X \to B$ and
  $p^X$ is postcomposition with $p$.
\end{lemma}

By the pullback lemma~\cite[Exercise~III.4.8]{MacLane:CWM}, we derive the following corollary. 

\begin{corollary}
  \label{cor:closed-under-factoring}
  If the morphisms $p \comp q$ and $p$ are internally right orthogonal to $X$ then so is $q$.
\end{corollary}

\begin{lemma}\label{lem:A:to:yoneda}
  Let $A$ be an object of $\BaseCat$. Let $\lambda$ be a clock and $n \in \NN$. As in Lemma~\ref{lem:object-of-clocks-is-colimit}
  we write simply $n$ for the map $\{\lambda\} \to \NN$ mapping $\lambda$ to $n$. 
  Then
  \[
    A^{y\left(\{\lambda\},n\right)} (\FSA, \delta) \iso A((\FSA,\fresh\FSA), \delta[\fresh\FSA\mapsto n])
  \]
  and up to this isomorphism, $c_A = A(\iota)$, where $\iota : (\FSA, \delta) \to ((\FSA,\fresh\FSA), \delta[\fresh\FSA\mapsto n])$
  is the inclusion.
\end{lemma}

\begin{proof}
In $\TimeCat$, the object $((\FSA,\fresh\FSA), \delta[\fresh\FSA\mapsto n])$ is a coproduct of $(\FSA, \delta)$ and 
$(\{\lambda\},n)$ with coproduct inclusions given by set inclusions (mapping $\lambda$ to $\fresh\FSA$). 
Since $y : \opcat \TimeCat \to \BaseCat$ 
preserves products, we get the following series of isomorphisms using the Yoneda lemma and standard definitions
of exponentials in presheaf categories:
\begin{align*}
 A^{y\left(\{\lambda\},n\right)} (\FSA, \delta) & = \Hom(y(\FSA,\delta)\times y\left(\{\lambda\},n\right),A) \\
 & \iso \Hom(y((\FSA,\fresh\FSA), \delta[\fresh\FSA\mapsto n]),A) \\
 & \iso A((\FSA,\fresh\FSA), \delta[\fresh\FSA\mapsto n])
\end{align*}

The morphism $c_A$ maps $x\in A(\FSA, \delta)$ to the natural transformation given by the composition of
the first projection $y(\FSA,\delta)\times y\left(\{\lambda\},n\right) \to y(\FSA,\delta)$ and the morphism 
$y(\FSA,\delta) \to A$ corresponding to $x$ under the Yoneda lemma. Since the projection corresponds 
to composition with $\iota$, the second statement of the lemma follows. 
\end{proof}

\begin{proofof}{Lemma~\ref{lem:fibrations-iff-restrictions-pullback}}
  \label{proof:lem:fibrations-iff-restrictions-pullback}
  Follows from Lemma~\ref{lem:characterization-of-fibrations} and Lemma~\ref{lem:A:to:yoneda}.
\end{proofof}

When interpreting syntax dependent types will be interpreted as families invariant under clock 
introduction. This will be used to prove soundness of the clock irrelevance axiom.
In fact, just to prove that, it would be enough that the interpretation of every type is 
internally right orthogonal to $\Cl$.
This is a slightly weaker statement than being invariant under clock introduction, as the next lemma states.
We have chosen to work with the latter because of the natural characterisation of Lemma~\ref{lem:fibrations-iff-restrictions-pullback}.

\begin{lemma} \label{lem:orthogonality-colimits}
Let $\CC$ be a cartesian closed category $\CC$ and let $X = \colim_i X_i$ be a connected colimit. 
If $p : A \to B$ is internally right orthogonal to all $X_i$, then it is also internally right orthogonal to $X$. 
As a consequence, if $p$ is invariant under clock introduction, it is also internally right orthogonal to $\Cl$. 
\end{lemma}

The second statement of the lemma follows from the first by Lemma~\ref{lem:object-of-clocks-is-colimit}.

The notion of internal orthogonality can be shown to be equivalent to the one used by 
\citeasnoun{Hyland:discrete-objects}, and the next lemma follows from \cite[Proposition~$2.1$]{Hyland:discrete-objects}.
Rather than proving this equivalence, we give here a direct proof.

\begin{proposition}
  \label{prop:orthogonality:closed}
  Suppose $\CC$ is a locally cartesian closed category and $X$ is an object in $\CC$.
  The notion of being internally right orthogonal to $X$ is then closed under composition, pullback (along \emph{arbitrary} maps), dependent products (along \emph{arbitrary} maps) and all isomorphisms are internally right orthogonal to $X$.
\end{proposition}

\begin{proof}
 Closure under composition and the fact that isomorphisms are internally right orthogonal to $X$ follow straightforwardly from Lemma~\ref{lem:characterization-of-fibrations}.
 
 To prove the statement for pullbacks, suppose $p : B \to D$ is internally right orthogonal to $X$ and $q : A\to C$ is the pullback of $p$ along
 some map $g$ not assumed to be internally right orthogonal to $X$. By the pullback pasting lemma then the outer square below is a pullback.
  \begin{displaymath}
    \begin{diagram}
      \pullbacktip
      A \ar{r}{f} \ar{d}[swap]{q} & \pullbacktip B \ar{d}{p} \ar{r}{c_B} & B^X \ar{d}{p^X}\\
      C \ar{r}{g} & D \ar{r}{c_D} & D^X
    \end{diagram}
  \end{displaymath}
  By naturality of $c$, the below outer square is equal to the one above, and thus also a pullback.
  \begin{equation}
    \label{eq:fibrations-pullback-two-pullbacks}
    \begin{diagram}
      A \ar{r}{c_A} \ar{d}[swap]{q} & A^X \ar{d}{q^X} \ar{r}{f^X} & B^X \ar{d}{p^X}\\
      C \ar{r}{c_C} & C^X \ar{r}{g^X} & D^X
    \end{diagram}
  \end{equation}
  Since $-^X$ has a left adjoint it preserves pullbacks and so right square of (\ref{eq:fibrations-pullback-two-pullbacks})
  is a pullback. By the pullback lemma, also the left square is a pullback, and thus $q$ is internally right orthogonal to $X$. 

For dependent products, suppose 
$p : A \to B$ is internally right orthogonal to $X$, and $f : B \to C$. We must show that $\Pi_f(p)$ is internally right orthogonal to $X$, where $\Pi_f : \CC/B \to \CC/C$ is the right adjoint to pullback along $f$. We write $f^*(h) : B \times_C Y \to B$ for the result of applying the 
pullback functor to an object $h : Y \to C$ of $\CC/C$ and use the notation
\[
 \widehat{(-)} : \Hom_{\ccat/C}(h,\Pi_f(p)) \to \Hom_{\ccat/B}(f^*(h),p)
\]
for the isomorphism of hom-sets, given $h : Y \to C$. 

Given $Y, h, k$ as in the outer square on the left below, by naturality, the isomorphism $\widehat{(-)}$ extends to 
a bijective correspondence of diagonal fillers in the following two diagrams. 
\begin{equation} \label{eq:two:filler:diagrams}
    \begin{diagram}
      Y \times X \ar{r}{k} \ar{d}[swap]{\pi_Y} & \Pi_BA\ar{d}{\Pi_f(p)}\\
      Y \ar{r}{h} \ar[dotted]{ur} & C
    \end{diagram} \GAP
    \begin{diagram}
      B\times_C (Y \times X)  \ar{r}{\widehat k} \ar{d}[swap]{B\times_C \pi_Y} & A\ar{d}{p}\\
      B\times_C Y \ar{r}{f^*(h)} \ar[dotted]{ur} & B
    \end{diagram}
\end{equation}
where $B\times_C \pi_Y$ is the pullback functor applied to the morphism $\pi_Y : (h\comp\pi_Y) \to h$ in $\CC/C$. 

By the pullback pasting lemma, the following outer diagram is a pullback
  \begin{displaymath}
    \begin{largediagram}
      \pullbacktip
      (B \times_C Y) \times X \ar{d}[swap]{\pi_{(B \times_C Y)}} \ar{r}{\pi_Y^C \times \id{X}} & Y \times X \ar{d}{\pi_Y}\\
      \pullbacktip
      B \times_C Y \ar{r}{\pi^C_Y} \ar{d}[swap]{f^*(h)} & Y\ar{d}{h}\\
      B \ar{r}{f} & C
    \end{largediagram}
  \end{displaymath}
From this we conclude that there is an isomorphism $\phi: (B \times_C Y) \times X \iso B\times_C(Y \times X)$. 
An easy diagram chase verifies (using the universal property of the lower diagram above) that 
\[(B\times_C\pi_{Y})\comp \phi = \pi_{B \times_C Y} : (B \times_C Y) \times X \to (B \times_C Y). \]
Thus, the fillers of (\ref{eq:two:filler:diagrams}) are in bijective correspondence with the fillers of 
\[
    \begin{diagram}
      (B\times_C Y)\times X \ar{r}{\widehat k\circ \phi} \ar{d}[swap]{\pi_{B\times_C Y}} & A\ar{d}{p}\\
      B\times_C Y \ar{r}{f^*(h)} \ar[dotted]{ur} & B
    \end{diagram}
\]
Since $p$ is assumed to be internally right orthogonal to $X$, there is a unique filler of the diagram above, and thus a
unique filler of the left diagram of (\ref{eq:two:filler:diagrams}). This proves that $\Pi_f(p)$ is internally right orthogonal to $X$ as desired. 
\end{proof}

\begin{corollary} \label{cor:orthogonality:closed}
 In the CwF structure of $\BaseCat$, the collection of families invariant under clock introduction is closed under 
 the operations for taking $\Pi$-, and $\Sigma$-types as well as reindexing. 
\end{corollary}

\begin{lemma}
  \label{lem:later-invariant-under-clock-introduction}
  If $\kappa \in \ElCwF{\Gamma}{\clockfam[\uniquemap_{\Gamma}]}$ and 
  the family $A \in \Fam\Gamma$ is invariant under clock introduction then 
  $\laterfam{\kappa}A$ is invariant under clock introduction.
\end{lemma}
\begin{proof}
 The map $\iota\cdot (-)$ is defined to be the identity on $\{\star\}$ in the case of $\delta(\kappa(\gamma)) = 0$. In the 
 case of $\delta(\kappa(\gamma)) > 0$ it is defined as the action 
 $\minkappa\sigma{\kappa(\gamma)}(-)$ on $A$. By the assumption the latter is always an 
 isomorphism and thus so is $\iota\cdot(-)$. 
\end{proof}

We now show that invariance under clock introduction implies the soundness of the clock irrelevance axiom.
\begin{lemma}
  \label{lem:clock-irrelevance-sound}
  Suppose $A$ in $\Fam{\Gamma}$ is invariant under clock introduction and that $t$ is in 
  $\ElCwF{\Gamma}{\PiSem{\clockfam[\uniquemap_\Gamma]}{A[\p]}}$ 
  and $\kappa, \kappa' \in \ElCwF{\Gamma}{\clockfam[\uniquemap_{\Gamma}]}$. Then $\ev(t, \kappa) = \ev(t, \kappa')$.
\end{lemma}
\begin{proof}
  Note first that $\ev(t[\p], \q) \in \ElCwF{\compr{\Gamma}{\clockfam[\uniquemap_{\Gamma}]}}{A[\p]}$. Since
  $\compr{\Gamma}{\clockfam[\uniquemap_{\Gamma}]} = \Gamma\times \clocks$, this gives us the commutative
  outer diagram below.
  \begin{displaymath}
    \begin{diagram}
      \Gamma\times \clocks \ar{rr}{\cpair{\p}{\ev(t[\p], \q)}} \ar{d}{\p} && \compr\Gamma A \ar{d}{\p}\\
      \Gamma \ar[dotted]{urr}{u}  \ar{rr}{\id{}} && \Gamma
    \end{diagram}
  \end{displaymath}
  Since  $\p : \compr\Gamma A \to \Gamma$ is internally right orthogonal to $\clocks$ by Lemma~\ref{lem:orthogonality-colimits}, 
  there is a unique lifting $u$ as indicated in the diagram.
  
  Now, 
\begin{align*}
 \cpair{\p}{\ev(t[\p], \q)} \circ \cpair{\id{}}\kappa & = \cpair{\p \circ\cpair{\id{}}\kappa}{\ev(t[\p], \q)[\cpair{\id{}}\kappa]} \\
 & =   \cpair{\id{}}{\ev(t[\p][\cpair{\id{}}\kappa], \q[\cpair{\id{}}\kappa])} \\
 & =   \cpair{\id{}}{\ev(t, \kappa)} 
\end{align*}
  Since $\cpair{\p}{\ev(t[\p], \q)} = u\circ\p$ this implies
  \[
  \cpair{\id{}}{\ev(t, \kappa)}  = u\circ \p \circ \cpair{\id{}}\kappa = u
  \]
  Likewise we can prove that $\cpair{\id{}}{\ev(t, \kappa')} = u$ and so $\ev(t, \kappa) = \ev(t, \kappa')$.
\end{proof}

\section{Identity types}
\label{sec:identity:types}

Since $\BaseCat$ is a presheaf category it models \emph{extensional} identity types, i.e., identity types with the identity 
reflection axiom. Recall that the rules for these are
\begin{mathpar}
 \inferrule*{
  \hastype{\Gamma}tA \and \hastype{\Gamma}uA
 }{\wftype{\Gamma}{\idty Atu}
 } \and
 \inferrule*{\hastype\Gamma tA
 }{\hastype{\Gamma}{\refl A(t)}{\idty Att}
 }\and
 \inferrule*{\hastype{\Gamma}p{\idty Atu}
 }{\judgeeq \Gamma tu
 }
\end{mathpar}
In the CwF structure, this structure is defined, for $t,u \in \ElCwF{\Gamma}A$ as 
\[
\idfam{A}tu_{(\FSA,\delta)}(\gamma) = \{\star \mid t(\gamma) = u(\gamma)  \}
\]

\begin{lemma}
 Let $A\in \Fam\Gamma$ and $t,u\in \ElCwF\Gamma A$. If $A$ is invariant under clock introduction, so is $\idfam{A}tu$.
\end{lemma}

\begin{proof}
 We must show that if  $\iota : (\FSA, \delta) \to \left((\FSA,\lambda), \delta[\lambda\mapsto n]\right)$
 is given by the inclusion, then 
 \[
 \iota\cdot (-) : \{\star \mid t(\gamma) = u(\gamma)  \} \to \{\star \mid t(\iota\cdot \gamma) = u(\iota\cdot \gamma)  \}
 \]
 is an isomorphism. 
 First recall that since $t$ and $u$ are elements, $t(\iota\cdot \gamma) = \iota\cdot t(\gamma)$ and likewise for $u$. 
 Since $A$ is invariant under clock introduction, $\iota\cdot(-)$ is an isomorphism on $A$, and so $ t(\gamma) = u(\gamma)$
 if and only if $\iota\cdot t(\gamma) = \iota\cdot u(\gamma)$. This implies that $\iota\cdot(-)$ on $\idfam{A}tu$
 is also an isomorphism as required. 
\end{proof}

\section{Delayed substitutions}
\label{sec:delayed:subst}

In the simply typed setting the applicative functor~\cite{McBride:Applicative} structure of the later modality is essential.
For instance, it allows us to apply a term $f$ of type $\later{\kappa}(A \to B)$ to a term $t$ of type $\later{\kappa}{A}$ to get a term $f \delayapp{\kappa} t$ of type $\later{\kappa}{B}$; that is, if we have a function after one $\kappa$-step and if after one $\kappa$-step we have an argument, we can apply the function at the time, and get the result after one $\kappa$-step.

In $\gdtt$ the function types can be dependent, and thus to be able to use the later modality to its fullest, the applicative functor structure needs to be generalised, so that we can apply a term $f$ of type $\later{\kappa}{\left(\depprod{x}{A}{B}\right)}$ to a term $t$ of type $\later{\kappa}{A}$.
In $\gdtt$ the type of the delayed application $f \delayapp{\kappa} t$ becomes $\later{\kappa}[\hrt{x \gets t}]{B}$ where $\hrt{x \gets t}$ is a \emph{delayed substitution}, and $x$ is bound in $\later{\kappa}[\hrt{x \gets t}]{B}$.
If at some point we learn that $t$ is of the form $\next\kappa{t'}$ for some $t'$ we can actually perform the substitution and get the type $\later{\kappa}{B[t'/x]}$.
This process can be iterated, e.g., if $B$ is also a dependent product $\depprod{y}{C}{D}$ and $s$ is a term of type $\later{\kappa}[\hrt{x \gets t}]C$ then the delayed application $f \delayapp{\kappa} t \delayapp{\kappa} s$ is well-typed with type $\later{\kappa}[\hrt{x \gets t, y \gets s}]D$.

Delayed substitutions satisfy convenient judgemental equalities (listed in Figure~\ref{fig:eq-rules-delayed-subst}) which ensure that delayed substitutions can be manipulated in an intuitive way.
For example, if the type $A$ is well-formed without $x$ then the delayed substitution in $\later{\kappa}{\hrt{x \gets t}}A$ is redundant, and thus $\later{\kappa}{\hrt{x \gets t}}A \jeq \later{\kappa}{A}$.
Further, as explained above, if the term $t$ is of type $\next{\kappa}t'$ then we can perform an actual substitution, and thus $\later{\kappa}[\hrt{x \gets \next{\kappa}{t'}}]{B} \jeq \later{\kappa}{B[t'/x]}$.
Finally, the order of bindings in $\later{\kappa}[\hrt{x \gets t, y \gets s}]D$ matters only in as much as it usually does in dependent type theory.
That is, $x \gets t$ and $y \gets s$ can be exchanged provided $x$ does not appear in the type of $y$.

To conclude this introduction to delayed substitutions we remark that they can be attached to the term former $\next{\kappa}t$ as well and they enjoy analogous rules.
As shown in previous work~\cite{Bizjak-et-al:GDTT} a calculus with just these generalised $\next{\kappa}{}$ and $\later{\kappa}{}$ can express the delayed application construct which was primitive in simply typed calculi with guarded recursion.
We refer to~\cite{Bizjak-et-al:GDTT} for extensive examples of how to use delayed substitutions for reasoning about guarded recursive and coinductive terms.

The typing rules for delayed substitutions and related constructs are recalled in Figure~\ref{fig:typing-rules-for-delayed-subst} and the equality rules are recalled in Figure~\ref{fig:eq-rules-delayed-subst}.
We write $\dsubst{\kappa}{\xi}{\Gamma}{\Gamma'}$ for the delayed substitution $\xi$ from $\Gamma$ to $\Gamma'$.
Note that $\Gamma'$ is not a context, but a telescope such that $\Gamma, \Gamma'$ is a well-formed context.
The delayed substitution $\xi$ is a list of pairs written as $x \gets t$, which are successively well-typed in context $\Gamma$ of types derived from $\Gamma'$, as stated in the formation rule in Figure~\ref{fig:typing-rules-for-delayed-subst}.

\begin{figure}[ht]
  \textbf{Delayed substitutions}
  \begin{mathpar}
    \inferrule{%
      \wfctx{\Gamma}\and \hastype{\Gamma}{\kappa}{\clocktype}}{%
      \dsubst{\kappa}{\cdot}{\Gamma}{\cdot}}
    \and
    \inferrule{%
      \dsubst{\kappa}{\xi}{\Gamma}{\Gamma'} \and
      \wftype{\Gamma, \Gamma'}A \and
      \hastype{\Gamma}{t}{\later{\kappa}[\xi]{A}}}{%
      \dsubst{\kappa}{\xi\hrt{x \gets t}}{\Gamma}{\Gamma', x:A}}
  \end{mathpar}

  \textbf{Well-formed types}
  \begin{mathpar}
    \inferrule{%
      \wftype{\Gamma,\Gamma'}{A} \\
      \dsubst{\kappa}{\xi}{\Gamma}{\Gamma'}}{%
      \wftype{\Gamma}{\later{\kappa}[\xi]{A}}}
  \end{mathpar}

  \textbf{Well-typed terms}
  \begin{mathpar}
    \inferrule{%
      \hastype{\Gamma,\Gamma'}{t}{A} \\
      \dsubst{\kappa}{\xi}{\Gamma}{\Gamma'}}{%
      \hastype{\Gamma}{\next{\kappa}[\xi]{t}}{\later{\kappa}[\xi]{A}}}
  \end{mathpar}
  \caption{Typing rules involving delayed substitutions.}
  \label{fig:typing-rules-for-delayed-subst}
\end{figure}

\begin{figure}[ht]
  \textbf{Type equality}
  \begin{mathpar}
    \inferrule{%
      \dsubst{\kappa}{\xi\hrt{x\gets t}}{\Gamma}{\Gamma',x:B} \\
      \wftype{\Gamma,\Gamma'}{A}}{%
      \eqjudg{\Gamma}{\later{\kappa}[\xi\hrt{x \gets t}]{A}}{\later{\kappa}[\xi]{A}}}
    \and
    \inferrule{%
      \dsubst{\kappa}{\xi\hrt{x\gets t,y\gets u}\xi'}{\Gamma}{\Gamma',x:B,y:C,\Gamma''} \and
      \wftype{\Gamma,\Gamma'}{C} \and
      \wftype{\Gamma,\Gamma',x:B,y:C,\Gamma''}{A}}{%
      \eqjudg{\Gamma}{\later{\kappa}[\xi\hrt{x\gets t,y\gets u}\xi']{A}}
      {\later{\kappa}[\xi\hrt{y\gets u,x\gets t}\xi']{A}}}
    \and
    \inferrule{%
      \dsubst{\kappa}{\xi}{\Gamma}{\Gamma'} \\
      \wftype{\Gamma,\Gamma', x:B}{A} \\
      \hastype{\Gamma,\Gamma'}{t}{B}}{%
      \eqjudg{\Gamma}{\later{\kappa}[\xi\hrt{x\gets \next{\kappa}[\xi]{t}}]{A}}
      {\later{\kappa}[\xi]{A\subst{x}{t}}}}
    \and
    \inferrule{
    \wftype{\Gamma, \Gamma', \Gamma''}{A} \\
    \dsubst{\kappa}{\xi}{\Gamma}{\Gamma'} \\
    \dsubst{\kappa}{\xi'}{\Gamma}{\Gamma''}}
    {
    \eqjudg{\Gamma}{\later{\kappa}[\xi]{\later{\kappa}[\xi'] A}}{\later{\kappa}[\xi']{\later{\kappa}[\xi] A}}
    }
    \and
    \inferrule{
     \dsubst{\kappa}{\xi}{\Gamma}{\Gamma'} \and \hastype{\Gamma,\Gamma'}{t}{A} \and \hastype{\Gamma,\Gamma'}{s}{A}
    }{
    \judgeeq{\Gamma}
  {\idty{\later{\kappa}[\xi]{A}}{\next{\kappa}[\xi]{t}}{\next{\kappa}[\xi]{s}}}
  {\later{\kappa}[\xi]{\idty{A}{t}{s}}}
  }
  \end{mathpar}
  \textbf{Term equality}
  \begin{mathpar}
    \inferrule{%
      \dsubst{\kappa}{\xi\hrt{x\gets t}}{\Gamma}{\Gamma', x:B} \\
      \hastype{\Gamma,\Gamma'}{u}{A} }{%
      \eqjudg{\Gamma}{%
        \next{\kappa}[\xi\hrt{x\gets t}]{u}}{%
        \next{\kappa}[\xi]{u}}[%
      \later{\kappa}[\xi]{A}]}
    \and
    \inferrule{%
      \dsubst{\kappa}{\xi\hrt{x\gets t,y\gets u}\xi'}{\Gamma}{\Gamma',x:B,y:C,\Gamma''} \and
      \wftype{\Gamma,\Gamma'}{C} \and
      \hastype{\Gamma,\Gamma',x:B,y:C,\Gamma''}{v}{A} }{%
      \eqjudg{\Gamma}{%
        \next{\kappa}[\xi\hrt{x\gets t,y\gets u}\xi']{v}}{%
        \next{\kappa}[\xi\hrt{y\gets u,x\gets t}\xi']{v}}[%
      \later{\kappa}[\xi\hrt{x\gets t,y\gets u}\xi']{A}]}
    \and
    \inferrule{%
      \dsubst{\kappa}{\xi}{\Gamma}{\Gamma'} \\
      \hastype{\Gamma,\Gamma',x:B}{u}{A} \\
      \hastype{\Gamma,\Gamma'}{t}{B}}{%
      \eqjudg{\Gamma}%
      {\next{\kappa}[\xi\hrt{x \gets \next{\kappa}[\xi]{t}}]{u}}%
      {\next{\kappa}[\xi]{u\subst{x}{t}}}%
      [\later{\kappa}[\xi]{A\subst{x}{t}}] }
    \and
    \inferrule{%
      \hastype{\Gamma}{t}{\later{\kappa}[\xi]{A}} }{%
      \eqjudg{\Gamma}%
      {\next{\kappa}[\xi\hrt{x\gets t}]{x}}%
      {t}%
      [\later{\kappa}[\xi]{A}]} 
    \and
    \inferrule{
    \wftype{\Gamma, \Gamma', \Gamma''}{A} \\
    \hastype{\Gamma, \Gamma', \Gamma''}{u}{A} \\
    \dsubst{\kappa}{\xi}{\Gamma}{\Gamma'} \\
    \dsubst{\kappa}{\xi'}{\Gamma}{\Gamma''}}
    {
    \eqjudg{\Gamma}{\next{\kappa}[\xi]\next{\kappa}[\xi']u}{\next{\kappa}[\xi']\next{\kappa}[\xi]u}[\later{\kappa}[\xi]{\later{\kappa}[\xi'] A}]
    }
  \end{mathpar}
  \caption{Equality rules involving delayed substitutions.}
  \label{fig:eq-rules-delayed-subst}
\end{figure}

The typing rule for $\prev{}$ is generalised in~\cite{Bizjak-et-al:GDTT} to allow elimination also of $\later{}$ with attached
delayed substitutions. We now recall that rule and show that it is admissible. 

\begin{proposition} \label{prop:generalised:prev}
 For any delayed substitution $\dsubst{\kappa}{\xi}{\Gamma, \kappa}{\Gamma'}$, there is a substitution 
\[\adv{\kappa}\xi : 
\Gamma, \kappa : \clocktype \to \Gamma, \kappa : \clocktype, \Gamma'\] 
defined as 
\begin{align*}
  \adv{\kappa}{(\cdot)} &= \id{\Gamma, \kappa}\\
  \adv{\kappa}{(\xi'[x \mapsto s])} &= (\adv{\kappa}{\xi'})[x \mapsto{\alwaysapp{(\prev\kappa. s)}\kappa}]
\end{align*}
such that whenever  $\wftype{\Gamma, \kappa, \Gamma'}A$ and $\hastype{\Gamma,\kappa : \clocktype}{t}{\later\kappa[\xi] A}$
also {$\hastype{\Gamma}{\prev\kappa . t}{\forall\kappa . A(\adv{\kappa}\xi)}$}.
Moreover, the following equality rule holds
\[
\inferrule*{
\hastype{\Gamma, \kappa, \Gamma'}{u}{A} \\
\dsubst{\kappa}{\xi}{(\Gamma, \kappa : \clocktype)}{\Gamma'}
}{
\eqjudg{\Gamma}{\prev\kappa . \next\kappa[\xi] u}{\Lambda\kappa .u(\adv\kappa(\xi))}[\forall\kappa . A(\adv\kappa(\xi))]
}
\]
\end{proposition}

Note that in the definition of $\adv{\kappa}{(\xi'[x \mapsto s])}$, the typing assumption on $s$ is 
\[
\hastype{\Gamma}{s}{\later\kappa[\xi'] A}
\]
and so the typing of $\adv{}$ relies on the second statement of the proposition. Thus the statements of welltypedness of $\adv\kappa(\xi)$
and of $\prev\kappa . t$ must be proved by simultaneous induction over the length of $\xi$. 

\begin{proof}
 We first define the concept of applying $\nxt\kappa$ to a substitution obtaining a delayed substitution. This should be thought of
 as an inverse operation to advancing a delayed substitution. Let $\sigma : \Gamma \to \Gamma, \Gamma'$ be a substitution
 which restricted to the context $\Gamma$ is the identity. Define
 \[
 \dsubst{\kappa}{\next\kappa(\sigma)}{\Gamma}{\Gamma'}
 \]
 by induction on the size of $\Gamma'$ by 
 \[
 \next\kappa\left(\sigma[x \mapsto u]\right) = \next\kappa(\sigma)[x \gets \next\kappa(u)]
 \]
 This is welltyped, since by assumption $\hastype\Gamma u{A\sigma}$ and so 
 $\hastype\Gamma {\next\kappa(u)}{\later\kappa (A\sigma)}$ 
 and 
 \[\later\kappa (A\sigma) \jeq \later\kappa[\next\kappa(\sigma)] A.\] 
 Since $\next\kappa(\adv\kappa (\xi)) \jeq\xi$ by the $\eta$ rule for $\prev{}$ the assumed type of $t$ in the statement of 
 the proposition is 
\begin{align*}
 \later\kappa[\xi] A & \jeq \later\kappa[\next\kappa(\adv\kappa (\xi))] A \\
 & \jeq \later\kappa A(\adv\kappa (\xi))
\end{align*}
by repeated application of the first and third rule of Figure~\ref{fig:eq-rules-delayed-subst}. 
Thus, $\hastype{\Gamma}{\prev\kappa . t}{\forall\kappa . A(\adv{\kappa}\xi)}$ as desired.
The equality rule stated at the end of the proposition follows analogously as
\begin{align*}
  \prev\kappa . \next\kappa[\xi] u
  &\jeq \prev\kappa . \next\kappa[\hrt{\next\kappa(\adv\kappa(\xi))}] u\\
  &\jeq \prev\kappa . \next\kappa \left(u(\adv\kappa(\xi))\right)\\
  &\jeq \Lambda\kappa .u(\adv\kappa(\xi))
\end{align*}
where the last equality is the $\beta$ rule for $\prev{}$ from Figure~\ref{fig:basic:syntax}.
\end{proof}

\subsection{Semantics of delayed substitutions.}

Let $\Gamma$ be an object of $\BaseCat$. A telescope over $\Gamma$ is a sequence of families $(A_1, \dots A_n)$, 
such that $A_{i+1} \in \Fam{\compr{\compr{\Gamma}{A_1\dots}}{A_i}}$ for each $i$. Let 
$\kappa \in \ElCwF{\Gamma}{\clockfam[\uniquemap_{\Gamma}]}$. We define the sets of
\emph{delayed sequence of elements} $\delayedEl{\Gamma}{\Gamma'}\kappa$ to be the set of mappings $\xi$,
associating to each $(\FSA,\delta)$ and $\gamma \in \Gamma_{(\FSA,\delta)}$ such that 
$\delta(\kappa(\gamma)) > 0$ a sequence $(\xi_1, \dots, \xi_n)$ such that 
\[
\xi_{i+1}(\gamma) \in A_{i+1}(\tick{\kappa(\gamma)}{\FSA}\cdot \gamma, \xi_1(\gamma), \dots, \xi_i(\gamma))
\]
and such that for every $\sigma : (\FSA,\delta) \to (\FSB,\delta')$ such that $\delta'(\kappa(\sigma\cdot \gamma)) > 0$
\[
\minkappa\sigma{\kappa(\gamma)}\cdot (\xi_{i+1}(\gamma)) = \xi_{i+1}(\sigma\cdot\gamma)
\]
Given a telescope $(A_1, \dots A_{n+1})$ over $\Gamma$ and  
$\xi \in \delayedEl{\Gamma}{(A_1, \dots A_{n})}\kappa$, define the family $\laterfam\kappa\xi . A_{n+1}$ over $\Gamma$
as
\begin{align*}
(\laterfam\kappa\xi . A_{n+1})(\gamma) & = 
  \begin{cases}
    \{\star\} & \text{ if } \delta(\kappa(\gamma)) = 0\\
    A_{n+1}(\tick{\kappa(\gamma)}{\FSA}\cdot \gamma, \xi_1(\gamma), \dots, \xi_n(\gamma)) & \text{ otherwise }
  \end{cases}
\end{align*}
with action $\sigma\cdot(-)$ defined using $\minkappa\sigma{\kappa(\gamma)}\cdot(-)$ on $A_{n+1}$. Note that this implies
that if $(\xi_1, \dots, \xi_n) \in \delayedEl{\Gamma}{(A_1, \dots, A_{n})}\kappa$ and $\xi_{n+1}$ is an element in 
$\laterfam\kappa\xi . A_{n+1}$ then  $(\xi_1, \dots, \xi_{n+1})$ is in the set $\delayedEl{\Gamma}{(A_1, \dots, A_{n+1})}\kappa$. 

If $t$ is an element
in $A_{n+1}$ define $\nextel\kappa\xi . t$ as an element of $\laterfam\kappa\xi . A_{n+1}$
by 
\begin{align*}
(\nextel\kappa\xi . t)(\gamma) & = 
  \begin{cases}
    \star & \text{ if } \delta(\kappa(\gamma)) = 0\\
    t(\tick{\kappa(\gamma)}{\FSA}\cdot \gamma, \xi_1(\gamma), \dots, \xi_n(\gamma)) & \text{ otherwise }
  \end{cases}
\end{align*}
If $\Gamma'$ is a telescope over $\Gamma$ and $\rho : \Delta \to \Gamma$, there is a telescope
$\Gamma'[\rho]$ over $\Delta$ and if further $\xi \in \delayedEl{\Gamma}{\Gamma'}\kappa$ 
we can define the reindexing $\xi[\rho] \in \delayedEl{\Delta}{\Gamma'[\rho]}{\kappa[\rho]}$ as 
$(\xi[\rho])_i(\gamma) = \xi_i(\rho(\gamma))$. The two above
constructions commute with reindexing in the sense that $(\laterfam\kappa\xi . A_{n+1})[\rho] =
\laterfam{\kappa[\rho]}\xi[\rho] . (A_{n+1}[\cpair{\rho\p^n}{\q[\p^{n-1}], \cdots, \q}])$ and likewise for $\nextel{}$. 

There are semantic correspondences to all of the syntactic equalities of Figure~\ref{fig:eq-rules-delayed-subst}, but
we only state and prove a few of these. 
We use notation similar to the syntax for delayed substitutions, e.g., if $\Gamma' = (A_1, \dots, A_n)$ is a telescope over 
$\Gamma$ we write $\compr{\Gamma}{\Gamma'}$ for $\compr{\compr{\Gamma}{A_1\dots}}{A_n}$. If 
$(\xi_1, \dots, \xi_n) \in \delayedEl{\Gamma}{(A_1, \dots, A_{n})}\kappa$ and $\xi_{n+1}$ is an element in 
$\laterfam\kappa\xi . A_{n+1}$ we write $\xi[\xi_{n+1}]$ for  $(\xi_1, \dots, \xi_{n+1})$.

\begin{theorem} \label{thm:later:del:subst:rules}
Let $\Gamma'$ be a telescope over $\Gamma$, $\kappa$ an element of $\clockfam[\uniquemap_{\Gamma}]$ 
and $\xi \in \delayedEl{\Gamma}{\Gamma'}\kappa$
\begin{enumerate}
\item \label{item:beta} If $t \in \ElCwF{\compr\Gamma{\Gamma'}}{B}$
 and $A \in \Fam{\compr{\compr\Gamma{\Gamma'}}B}$. Then 
 \[
 \laterfam\kappa\xi[\nextel\kappa\xi . t] . A =  \laterfam\kappa\xi . (A[\cpair{\id{}}t])
 \]
 \item If also $\Gamma''$ is a telescope over $\Gamma$, $\xi' \in \delayedEl{\Gamma}{\Gamma''}\kappa$ 
 and $A$ is in
 $\Fam{\compr{\compr\Gamma{\Gamma'}}{\Gamma''[\p]}}$ where 
 $\p : \compr\Gamma{\Gamma'} \to \Gamma$ then 
 \[
 \laterfam\kappa\xi . (\laterfam\kappa(\xi'[\p]) . A) = \laterfam\kappa\xi' . (\laterfam\kappa(\xi[\p]) . A[\mathsf{swap}])
 \]
 where $\mathsf{swap}: \compr{\compr\Gamma{\Gamma'}}{\Gamma''[\p]}
 \to \compr{\compr\Gamma{\Gamma''}}{\Gamma'[\p]}$ is the obvious map. 
 \item If $A$ is in $\Fam{\compr\Gamma{\Gamma'}}$ and 
 $t, u \in \ElCwF{\compr\Gamma{\Gamma'}}A$. Then
 \[
 \idfam{\laterfam\kappa \xi . A}{\nextel\kappa \xi .t}{\nextel\kappa \xi .u} = \laterfam\kappa \xi .(\idfam Atu)
 \]
 \end{enumerate}
\end{theorem}

\begin{proof}
 Write $\xi = (\xi_1, \dots, \xi_n)$. For the first one in the case of $\delta(\kappa(\gamma))>0$ we get  
\begin{align*}
 \laterfam\kappa\xi[\nextel\kappa\xi . t] . A(\gamma) 
 & = A(\tick{\kappa(\gamma)}{\FSA}\cdot \gamma, \xi_1(\gamma), \dots, \xi_n(\gamma), \nextel\kappa\xi . t(\gamma)) \\
 & = A(\tick{\kappa(\gamma)}{\FSA}\cdot \gamma, \xi_1(\gamma), \dots, \xi_n(\gamma), 
 t(\tick{\kappa(\gamma)}{\FSA}\cdot \gamma, \xi_1(\gamma), \dots, \xi_n(\gamma))) \\
 & = A[\cpair{\id{}}t](\tick{\kappa(\gamma)}{\FSA}\cdot \gamma, \xi_1(\gamma), \dots, \xi_n(\gamma)) \\
 & = \laterfam\kappa\xi . (A[\cpair{\id{}}t])(\gamma) 
\end{align*}
In the second one if $\delta(\kappa(\gamma))<2$ both sides are $\{\star\}$. Otherwise, writing $\xi(\gamma)$ for 
$(\xi_1(\gamma), \dots, \xi_n(\gamma))$ and likewise for $\xi'$ we get
\begin{align*}
 \laterfam\kappa\xi . (\laterfam\kappa(\xi'[\p]) . A)(\gamma) 
 & = (\laterfam\kappa(\xi'[\p]) . A)(\tick{\kappa(\gamma)}{\FSA}\cdot \gamma, \xi(\gamma)) \\
 & = A(\tick{\kappa(\gamma)}{\FSA}\cdot\tick{\kappa(\gamma)}{\FSA}\cdot \gamma, \tick{\kappa(\gamma)}{\FSA}\cdot\xi(\gamma),   \xi'(\tick{\kappa(\gamma)}{\FSA}\cdot\gamma)) \\
     & = A(\tick{\kappa(\gamma)}{\FSA}\cdot\tick{\kappa(\gamma)}{\FSA}\cdot \gamma, \xi(\tick{\kappa(\gamma)}{\FSA}\cdot\gamma),  \tick{\kappa(\gamma)}{\FSA}\cdot\xi'(\gamma)) \\
     & = A[\mathsf{swap}](\tick{\kappa(\gamma)}{\FSA}\cdot\tick{\kappa(\gamma)}{\FSA}\cdot \gamma, \tick{\kappa(\gamma)}{\FSA}\cdot\xi'(\gamma), \xi(\tick{\kappa(\gamma)}{\FSA}\cdot\gamma)) \\
  & =\laterfam\kappa(\xi[\p]) . A[\mathsf{swap}](\tick{\kappa(\gamma)}{\FSA}\cdot \gamma, \xi(\gamma)) \\
  & = \laterfam\kappa\xi' . (\laterfam\kappa(\xi[\p]) . A[\mathsf{swap}])(\gamma)
\end{align*}

In the last statement, if $\delta(\kappa(\gamma)) = 0$ both sides are $\{\star\}$. Otherwise 
\begin{align*}
 \idfam{\laterfam\kappa \xi . A}{\nextel\kappa \xi .t}{\nextel\kappa \xi .u}(\gamma) 
 & = \{\star\mid (\nextel\kappa \xi .t)(\gamma) = (\nextel\kappa \xi .u)(\gamma)  \} \\
 & = \{\star\mid t(\tick{\kappa(\gamma)}{\FSA}\cdot\gamma, \xi(\gamma)) = u(\tick{\kappa(\gamma)}{\FSA}\cdot\gamma, \xi(\gamma)) \} \\
 & = \idfam Atu(\tick{\kappa(\gamma)}{\FSA}\cdot\gamma, \xi(\gamma)) \\
 & = \laterfam\kappa \xi .(\idfam Atu)(\gamma)
\end{align*}
\end{proof}

Finally we note that the collection of families invariant under clock introduction is closed under $\laterfam{}$.

\begin{proposition}\label{prop:laterfam:inv:clock:intro}
 If $A$ is  invariant under clock introduction so is $\laterfam{\kappa}\xi . A$.
\end{proposition}

\begin{proof}
 The conclusion follows directly from the hypothesis since $\iota\cdot(-)$ on $\laterfam{\kappa}\xi . A$
 is defined to be $\minkappa\iota{\kappa(\gamma)}\cdot(-)$ as defined on $A$ when $\delta(\kappa(\gamma)) > 0$
 and the identity when $\delta(\kappa(\gamma)) = 0$. 
\end{proof}

\section{Universes}
\label{sec:universes}

We now assume we are given a set theoretic universe with its induced notion of small sets. Being a presheaf category, 
$\BaseCat$ has a universe object $\Usem{}$ and a dependent type $\Elsem{}$ of elements defined as 
in~\cite{Hofmann-Streicher:lifting}, as we now recall. If $(\FSA,\delta)$ is a 
time object, then the set $\Usem{}(\FSA,\delta)$ is the set of small families over $y(\FSA, \delta)$. Concretely, 
an element $X$ in $\Usem{}(\FSA,\delta)$ assigns to each $\sigma : (\FSA,\delta) \to (\FSB,\delta')$ a small set 
$X_\sigma$ and to each $\tau : (\FSB,\delta') \to (\FSC,\delta'')$ a map $\tau\cdot(-) : X_\sigma\to X_{\tau\sigma}$
in a functorial way. 
The action $\sigma\cdot (-) : \Usem{}(\FSA,\delta) \to \Usem{}(\FSB,\delta')$ maps an $X$ to the family $(X_{\tau\sigma})_\tau$.
The family $\Elsem{}$ is defined as $\Elsem{}_{(\FSA,\delta)}(X) = X_{\id{}}$ with action $\sigma\cdot(-) :
\Elsem{}_{(\FSA,\delta)}(X) \to \Elsem{}_{(\FSB,\delta')}(\sigma\cdot X)$ defined as $\sigma\cdot(-) : X_{\id{}} \to X_\sigma$. 

One might hope that this universe could be used to model an extension of $\coregdtt$ with one universe. 
However, 
$\Usem{}$ is not a constant presheaf, and therefore not invariant under clock introduction. Another way to see this is that
the map
\[
\later{} : \Cl\times \Usem{} \to \Usem{}
\]
defined, at $(\FSA,\delta)\in \BaseCat$, as 
\[
(\later{}(\lambda, X))_{\sigma : (\FSA, \delta) \to (\FSB,\delta')} = 
\begin{cases}
 1 & \text{ if } \delta'(\sigma(\lambda)) = 0 \\
 X_{\tick{\sigma(\lambda)}{} \circ \sigma} & \text{ else }
\end{cases}
\]
does not factor through the second projection. One can restrict the universe $\Usem{}$ 
to the families invariant under clock introduction, i.e., those $X$ such that $\iota \cdot (-) : X_\sigma \to X_{\iota\sigma}$
is an isomorphism for $\iota$ of the relevant form, but this 
does not rule out the problematic map, and so does not eliminate the problem. Note that $\later{}$ above does indeed encode 
the constructor $\laterfam{}$ since if $A \in \ElCwF{\Gamma}{\Usem{}[\uniquemap_{\Gamma}]}$ 
and $\kappa \in  \ElCwF{\Gamma}{\clockfam[\uniquemap_{\Gamma}]}$ then, 
if $\delta(\kappa(\gamma))>0$, 
\begin{align*}
\Elsem{}[(\later{}(\kappa, A))](\gamma) & = (\later{}(\kappa(\gamma), A(\gamma)))_{\id{}} \\
& = (A(\gamma))_{\tick{\kappa(\gamma)}{}} \\
& = (A(\tick{\kappa(\gamma)}{}\cdot \gamma))_{\id{}} \\
& = (\Elsem{}[A])(\tick{\kappa(\gamma)}{}\cdot \gamma) \\
& = \laterfam\kappa(\Elsem{}[A])(\gamma)
\end{align*}

To avoid this problem 
we follow the approach of $\gdtt$ and introduce, for each finite set of clock variables $\Delta$, a universe of types 
depending on the clocks in $\Delta$. An element in this universe is to be thought of as being constant in the dimensions
outside $\Delta$, and the operation $\later{\kappa}$ is only defined on the universe for $\kappa \in \Delta$. This rules
out the more general $\later{}$ operation mentioned above. 
It also means that universes are now indexed over a new dimension
(clock contexts). We show that the operations on types are polymorphic in this dimension. 

\subsection{Universes in \gdtt}

We first describe the syntax of Tarski style universes in \gdtt. The basic rules are listed in Figure~\ref{fig:universes}
and the rules for type operations on the universes are listed in Figure~\ref{fig:universes:dep:prod}. The type $\univ{\Delta}$ 
can be formed in a context $\Gamma$, whenever $\Delta$ is a sequence of clocks in that context, but the equality
rules say that the universes formed by two lists are equal if the lists contain the same elements. Inclusions between sets of
clocks induce inclusions between universes and these commute with taking types of elements as well as with all
type operations. This is the notion of universe polymorphism in the clock dimension referred to above. The universe
$\univ\Delta$ is closed under $\later\kappa$, but only for $\kappa \in \Delta$ thus avoiding the problem 
described above. The choice of domain type $\later\kappa\univ\Delta$
for $\latercode\kappa{}$ ensures that guarded recursive types can be defined by guarded recursion. For example, if
$B : \univ\Delta$ and $\kappa \in \Delta$ we can define a type of guarded recursive streams over $B$ as
\[
\gstreamcode\kappa B \defeq \fix\kappa A. B \timescode \,\latercode\kappa A : \univ\Delta 
\]
where $\timescode$ is encoded using $\Sigma$-types in the usual way. Then
\begin{align*}
 \elems\Delta(\gstreamcode\kappa B) 
 & \jeq \elems\Delta(B \timescode \,\latercode\kappa(\next\kappa(\gstreamcode\kappa B))) \\
 & \jeq \elems\Delta(B) \times \elems\Delta(\latercode\kappa(\next\kappa(\gstreamcode\kappa B))) \\
 & \jeq \elems\Delta(B) \times \later\kappa(\elems\Delta(\gstreamcode\kappa B)) 
\end{align*}
thus satisfying the expected type equality for guarded recursive streams over $\elems\Delta(B))$. 
In the last equality of Figure~\ref{fig:universes:dep:prod}, the typing assumption on $A$ is
$\hastype\Gamma A{\later\kappa\univ\Delta}$ and since 
${\hastype{\Gamma, x : \univ\Delta}{\univin{\Delta}{\Delta'}{(x)}}{\univ{\Delta'}}}$
also $\hastype\Gamma{\next{\kappa}[\hrt{x \gets A}]{\univin{\Delta}{\Delta'}{(x)}}}{\later\kappa \univ{\Delta'}}$

\begin{figure}[tbp]
\paragraph{Formation and typing rules}
\begin{mathpar}
 \inferrule*{
 \wfclock{\Gamma}{\kappa_1} \\ \dots \\ \wfclock{\Gamma}{\kappa_n}
 }
 {\wftype{\Gamma}{\univ{\kappa_1,\dots,\kappa_n}}} \and
 \inferrule*{
 \hastype{\Gamma}{t}{\univ{\Delta}}
 }
 {\wftype{\Gamma}{\elems\Delta(t)}} \and
 \inferrule*{
 \hastype{\Gamma}{t}{\univ{\Delta}} \\ \wftype{\Gamma}{\univ{\Delta'}} \\ \Delta\subseteq \Delta'
 }{
 \hastype{\Gamma}{\univin{\Delta}{\Delta'} (t)}{\univ{\Delta'}}}
\end{mathpar}
\paragraph{Equations}
\begin{align*}
 \univ\Delta & \jeq \univ{\Delta'} \GAP \text{if $\Delta =\Delta'$ as sets}  \\
 \elems{\Delta'}(\univin{\Delta}{\Delta'} (t)) & \jeq \elems{\Delta}(t) \\
  \univin{\Delta'}{\Delta''}(\univin{\Delta}{\Delta'}(t)) & \jeq \univin{\Delta}{\Delta''}(t)
\end{align*}
\begin{center}
\caption{Universes in \gdtt.}
\label{fig:universes}
\end{center}
\end{figure}

\begin{figure}[tbp]
\paragraph{Formation and typing rules}
\begin{mathpar}
 \inferrule*{
 \hastype{\Gamma}{A}{\univ{\Delta}} \\ 
 \hastype{\Gamma, x: \elems\Delta(A)}B{\univ{\Delta}}
 }
 {
  \hastype{\Gamma}{\depprodcode\Delta{x}{A}B}{\univ\Delta}
 }  \and
 \inferrule*{
 \hastype{\Gamma}{A}{\univ{\Delta}} \\ 
 \hastype{\Gamma, x: \elems\Delta(A)}B{\univ{\Delta}}
 }
 {\hastype{\Gamma}{\depsumcode\Delta{x}{A}B}{\univ\Delta}}  
 \and
  \inferrule*{
 \hastype{\Gamma, \kappa: \clocktype}A{\univ{\Delta,\kappa}}\\
 \kappa\notin\Delta
 }
 {
  \hastype{\Gamma}{\forallcode\kappa . A}{\univ\Delta}
 } \and
 \inferrule*
 {%
   \kappa \in \Delta%
   \and%
   \hastype{\Gamma}{A}{\later{\kappa}{\univ{\Delta}}}%
 }%
 {%
   \hastype{\Gamma}{\latercode{\kappa}{A}}{\univ{\Delta}}%
 }
\end{mathpar}
\paragraph{Equations}
\begin{align*}
 \elems\Delta(\depprodcode{\Delta}{x}{A}B) & \jeq \depprod x{\elems\Delta A}{\elems\Delta B} \\
 \elems\Delta(\depsumcode{\Delta}{x}{A}B) & \jeq \depsum x{\elems\Delta A}{\elems\Delta B} \\
 \elems\Delta(\forallcode\kappa . A) & \jeq \forall\kappa . \elems{\Delta,\kappa} (A) \\
 \elems{\Delta}\left(\latercode{\kappa}{\next{\kappa}[\xi]{A}}\right)
 & \jeq \later{\kappa}[\xi]\left(\elems{\Delta}{A}\right)\\
 \depprodcode{\Delta'}x{\univin{\Delta}{\Delta'} (A)}{\univin{\Delta}{\Delta'} (B)} 
 & \jeq \univin{\Delta}{\Delta'} (\depprodcode{\Delta}{x}{A}B) \\
 \depsumcode{\Delta'}x{\univin{\Delta}{\Delta'} (A)}{\univin{\Delta}{\Delta'} (B)} 
 & \jeq \univin{\Delta}{\Delta'} (\depsumcode{\Delta}{x}{A}B) \\
 \univin{\Delta}{\Delta'}{(\forallcode\kappa . A)} & \jeq \forallcode\kappa . \univin{(\Delta,\kappa)}{(\Delta',\kappa)}{(A)} \\
 \univin{\Delta}{\Delta'}{(\latercode{\kappa}{A})}
 & \jeq \latercode{\kappa}{(\next{\kappa}[\hrt{x \gets A}]{\univin{\Delta}{\Delta'}{(x)}})} 
\end{align*}
\begin{center}
\caption{Syntax for codes for basic operations on types}
\label{fig:universes:dep:prod}
\end{center}
\end{figure}

\subsection{A family of semantic universes}

To model the universe $\univ\Delta$, we must be in a context where $\Delta$ is defined, and the smallest
syntactic context where this happens is the one with $|\Delta|$-many clock variables. This is modelled by the $\BaseCat$
object $\Cl^\Delta$ defined as $\Cl^\Delta(\FSA,\delta) = \FSA^{\Delta}$. Note here that we treat $\Delta$ as a set, and so
the exponential $\FSA^{\Delta}$ is the ordinary set-theoretic one. The universe $\univ\Delta$ will be modelled as 
a family $\Usem\Delta$ over $\Cl^\Delta$. Recall that such a type corresponds to a (covariant) presheaf over the category of elements of 
$\Cl^\Delta$, i.e., the category whose elements are triples $(\FSA,\delta, f)$, such that the first two components
constitute an object in $\TimeCat$ and the last is a map $f : \Delta \to \FSA$. A morphism 
$\sigma : (\FSA,\delta, f) \to (\FSB,\delta', g)$ is a morphism $\sigma : (\FSA,\delta) \to (\FSB,\delta')$ in $\TimeCat$ such that 
$g = \sigma f$.
We will write $\gr\Delta$ for the category of covariant presheaves over this category, and use the same notation ($\gr\Delta$) for
the CwF-structure defined similarly to the CwF structure on $\BaseCat$. The semantic universe $\Usem\Delta$
will be an object in $\gr\Delta$ and the type $\elems\Delta$ will be modelled as a family $\Elsem\Delta$ over $\Usem\Delta$.

To avoid the problem described above with the standard universe in $\BaseCat$, the universe $\Usem\Delta$ 
should restrict access at level $(\FSA, \delta, f)$ to the clocks defined in $\Delta$. To do this, we define
$\Usem\Delta_{(\FSA, \delta, f)}$ to be the set of small families in $\BaseCat$ 
over $y(f[\Delta], \restrict{\delta}{f[\Delta]})$ invariant under clock introduction. Here $f[\Delta] \subseteq\FSA$ is the image
of $f$, and the notion of small families should be understood as described above. In other words, an element of 
$\Usem\Delta_{(\FSA, \delta, f)}$ is a family of sets 
$X_\sigma$ indexed over morphisms $\sigma$ in $\TimeCat$ with domain $(f[\Delta], \restrict{\delta}{f[\Delta]})$ 
together with maps $\tau\cdot (-) : X_\sigma \to X_{\tau\sigma}$ satisfying functoriality.
The requirement of invariance under clock introduction means that if 
$\iota : (\FSB,\delta') \to ((\FSB,\lambda),\delta'[\lambda \mapsto n])$ is an inclusion, then $\iota\cdot(-)$ must be an isomorphism.
 
For $\sigma : (\FSA,\delta, f) \to (\FSB,\delta', \sigma f)$ we must define $\sigma\cdot(-) : \Usem{\Delta}_{(\FSA,\delta, f)} \to \Usem{\Delta}_{(\FSB,\delta', \sigma f)}$.
Denote by $\cores\sigma{}$ the restriction and corestriction of $\sigma$:
\[
\cores \sigma{} : (f[\Delta], \restrict{\delta}{f[\Delta]}) 
\to (\sigma f[\Delta], \restrict{\delta'}{\sigma f[\Delta]}) \, .
\]
Using this, we define the family $(\sigma\cdot X)_\tau = X_{\tau\cores\sigma{}}$ 
for $\tau :  (\sigma f[\Delta], \restrict{\delta'}{\sigma f[\Delta]}) \to (\FSB, \delta'')$. 
Note that this is well-defined, i.e., if 
$X$ is invariant under clock introduction, so is $\sigma \cdot X$.

Since $\gr\Delta$ is equivalent to the slice of $\BaseCat$ over $\Cl^{\Delta}$, the notion of 
invariance under clock introduction extends to objects and families in $\gr\Delta$ by requiring 
the same for their corresponding projection maps in $\BaseCat$. By 
Lemma~\ref{lem:fibrations-iff-restrictions-pullback} this can be reformulated as requiring that 
the maps $\iota \cdot(-)$ induced by maps of the form 
$\iota : (\FSA,\delta, f) \to ((\FSA,\lambda),\delta[\lambda\mapsto n], \iota f)$ are
isomorphisms.

\begin{lemma} \label{lem:Usem:invariant:clock}
 The object $\Usem{\Delta}$ is invariant under clock introduction. 
\end{lemma}

\begin{proof}
  If $\iota : (\FSA,\delta, f) \to ((\FSA,\lambda),\delta[\lambda\mapsto n], \iota f)$ 
  is the inclusion then $\iota f[\Delta] = f[\Delta]$, and  $\cores{\iota}{}$ is the identity,
  so $(\iota\cdot X)_\tau = X_{\tau\cores\iota{}} = X_\tau$, i.e., $\iota\cdot(-)$ is the identity
  and therefore an isomorphim.
\end{proof}

If $X$ is an element in $\Usem{\Delta}_{(\FSA,\delta, f)}$, define 
\[
  \Elsem{\Delta}_{(\FSA,\delta, f)}(X) = X_{i : (f[\Delta], \restrict{\delta}{f[\Delta]}) \to (\FSA,\delta)}
\]
where $i$ is the inclusion. If $\sigma : (\FSA,\delta) \to (\FSB,\delta')$ we must define 
\[
  \sigma\cdot(-) : \Elsem{\Delta}_{(\FSA,\delta, f)}(X) 
  \to \Elsem{\Delta}_{(\FSB,\delta', \sigma f)}(\sigma\cdot X)
\]
The codomain of this map is 
\begin{align*}
 \Elsem{\Delta}_{(\FSB,\delta', \sigma f)}(\sigma\cdot X) & = 
(\sigma\cdot X)_{j} 
 = X_{j\circ \cores{\sigma}{}} 
 = X_{\sigma\circ i}
\end{align*}
where $j : \sigma[f[\Delta]] \to \FSB$ is the inclusion. 
We can therefore define $\sigma\cdot(-) : X_{i} \to X_{\sigma\circ i}$ to be the map that is part of the 
structure of $X$. 

\begin{lemma} \label{lem:Elsem:invariant:clock}
 The family $\Elsem{\Delta}$ over $\Usem{\Delta}$ is invariant under clock introduction. 
\end{lemma}

\begin{proof}
 Let $\iota : (\FSA,\delta, f) \to ((\FSA,\lambda),\delta[\lambda\mapsto n], \iota f)$ be the inclusion and
 let $X$ be an element in $\Usem{\Delta}_{(\FSA,\delta, f)}$. We must show that 
 \[\iota\cdot(-) :  \Elsem{\Delta}_{(\FSA,\delta, f)}(X) \to 
 \Elsem{\Delta}_{((\FSA,\lambda),\delta[\lambda\mapsto n], \iota f)}(\iota\cdot X)
 \]
 is an isomorphism. By definition of $\Elsem{\Delta}$ this map is $\iota\cdot(-) : X_i \to X_{\iota\circ i}$,
 which is part of the structure of $X$. Since $X$ is an element in the universe $\Usem{\Delta}$ it must be 
 invariant under clock introduction, which means exactly that all maps of the form $\iota\cdot(-)$ 
 are isomorphisms. 
\end{proof}

We now describe an abstract construction that leads to the universe $\Usem\Delta$. This construction will
not be used in the remainder of the paper, and so is not of technical importance, but perhaps of 
conceptual interest to some readers. Consider the universe $\Usem{}$ of small families invariant
under clock introduction in $\BaseCat$. There is a functor $F$ from $\BaseCat$ to $\gr\Delta$ 
mapping an object $\Gamma$ to
the presheaf whose value at $(\FSA, \delta, f)$ is $\Gamma(f[\Delta],\restrict\delta{f[\Delta]})$. This 
extends to families by mapping $A$ over $\Gamma$ to the family whose value at 
$\gamma \in \Gamma(f[\Delta],\restrict\delta{f[\Delta]})$ is $A(i\cdot \gamma)$. Note that this is not a mapping
of CwFs, since it does not preserve comprehension:
\begin{align*}
 F(\compr\Gamma A)(\FSA, \delta, f) 
 & = \{\pair \gamma a \mid \gamma \in \Gamma(f[\Delta],\restrict\delta{f[\Delta]}), a\in A(\gamma) \} \\
 \compr{F(\Gamma)}{F(A)}(\FSA, \delta, f) 
 & = \{\pair \gamma a \mid \gamma \in \Gamma(f[\Delta],\restrict\delta{f[\Delta]}), a\in A(i \cdot\gamma) \}
\end{align*}
(Although it does up to isomorphism if attention is restricted to 
families invariant under clock introduction). The universe $\Usem\Delta$
is this mapping applied to $\Usem{}$ and $\Elsem\Delta$ is the same mapping applied to the family of elements 
over $\Usem{}$. 

The next key lemma gives a partial answer to the question of what the universes $\Usem\Delta$
classify. The answer is partial, since it only applies in contexts invariant under clock introduction. 
As we shall see below, this result is sufficient for constructing codes for type operations on the universes.

\begin{lemma} \label{lem:Ufam:class}
 Let $\Gamma$ be an object in $\gr\Delta$ invariant under clock introduction and let $A$ be a small family over 
 $\Gamma$, also
 invariant under clock introduction. There is a unique code $\code A : \Gamma \to \Usem\Delta$ in $\gr\Delta$ such that
 $A = \Elsem\Delta[\code A]$.
\end{lemma}

\begin{proof}
 The assumption of invariance under clock introduction implies that for any object $(\FSA, \delta, f)$ the map $i \cdot(-)$
 induced by $i : (f[\Delta], \restrict\delta{f[\Delta]}, f) \to (\FSA, \delta, f)$ is an isomorphism on $\Gamma$. We will
 write $\inv i\cdot (-)$ for the inverse map. The code $\code A$ is defined as 
 \[
 (\code A_{(\FSA, \delta, f)}(\gamma))_{\tau : (f[\Delta], \restrict\delta{f[\Delta]}, f) \to (\FSB, \delta', \tau f)}
 = A_{(\FSB, \delta', \tau f)}(\tau \cdot \inv i\cdot \gamma)
 \]
 We first show that this defines a map of presheaves: If $\sigma : (\FSA, \delta, f) \to (\FSA', \delta', \sigma f)$ 
 and $\tau : (\sigma f[\Delta], \restrict{\delta'}{\sigma f[\Delta]}, \sigma f) \to (\FSC, \delta'', \tau \sigma f)$ then
\begin{align*}
 (\code A_{(\FSB, \delta', \sigma f)}(\sigma \cdot \gamma))_{\tau}
 & = A_{(\FSC, \delta'', \tau\sigma f)}(\tau \cdot \inv j\cdot \sigma \cdot \gamma)
\end{align*}
where $j : (\sigma f[\Delta], \restrict{\delta'}{\sigma f[\Delta]}, \sigma f)  \to (\FSB, \delta', \sigma f)$ is the inclusion.
Since $\sigma i = j \cores\sigma{}$ also $\inv j\cdot \sigma \cdot \gamma = \cores\sigma{} \cdot \inv i \cdot \gamma$
and so 
\begin{align*}
 (\code A_{(\FSB, \delta', \sigma f)}(\sigma \cdot \gamma))_{\tau}
 & = A_{(\FSC, \delta'', \tau\sigma f)}(\tau \cdot \cores\sigma{} \cdot \inv i\cdot \gamma) \\
  & = (\code{A}_{(\FSA, \delta, f)}(\gamma))_{\tau\cores\sigma{}} \\
  & = (\sigma\cdot(\code{A}_{(\FSA, \delta, f)}(\gamma)))_{\tau}
\end{align*}
so $(\code{A}(\sigma\cdot\gamma)) = \sigma\cdot(\code{A}(\gamma))$ meaning that $\code A$ is a map of presheaves.

This definition defines a code for $A$ since
\begin{align*}
 (\Elsem\Delta[\code A])(\gamma) & = \Elsem\Delta(\code A(\gamma)) 
  = (\code A(\gamma))_i 
  = A(i\cdot \inv i\cdot \gamma) 
  = A(\gamma)
\end{align*}
For uniqueness, suppose $\rho : \Gamma \to \Usem\Delta$ satisfies $\Elsem\Delta[\rho] = A$. We must show that 
$\rho(\gamma) = \code A(\gamma)$ for all $\gamma$, but consider first the case where $\gamma\in \Gamma_{(\FSA,\delta, f)}$
for $f$ surjective. In that case $(\code A(\gamma))_\tau = A(\tau\cdot \gamma)$ and
\begin{align*}
 (\rho(\gamma))_\tau & = (\rho(\gamma))_{j\overline\tau} 
  = (\tau\cdot \rho(\gamma))_j
  = (\rho(\tau \cdot \gamma))_j 
  = \Elsem\Delta(\rho(\tau\cdot\gamma)) 
  = A(\tau \cdot \gamma)
\end{align*}
where $j : \tau f[\Delta] \to \FSB$ is the inclusion. In general (when $f$ is not surjective) the above implies
\[
\rho(\gamma) = \rho(i\cdot \inv i\cdot \gamma) = i\cdot \rho(\inv i\cdot \gamma) = i\cdot \code A(\inv i\cdot \gamma) 
= \code A(i\cdot \inv i\cdot \gamma) = \code A(\gamma)
\]
proving uniqueness.
\end{proof}

\subsection{Reindexing universes}

The idea for interpreting the formation rule for the universes $\univ\Delta$ in a semantic
context $\Gamma$, is to interpret each $\kappa \in \Delta$ as an element of 
$\clockfam[\uniquemap_{\Gamma}]$, then use this to define a map from $\Gamma$
to $\Cl^\Delta$ in $\BaseCat$, and reindex the universe $\Usem\Delta$ along this map. 
The last of these steps uses the fact that an object of $\gr\Delta$ can be considered
a family of $\BaseCat$ over $\Cl^\Delta$. In fact, these two notions are equivalent. 

In order to prove the substitution lemma, we will generalise the above idea slightly 
as follows. Suppose $\clocksetA$ is a finite set of morphisms from $\Gamma$ to
$\Cl$, and suppose we are given a surjective map from some set $\Delta$ to $\clocksetA$
inducing a map $\clocksetAmap : \Gamma \to \Cl^\Delta$. Define 
\begin{align*}
\Usem{\clocksetA} & \eqdef \Usem\Delta[\clocksetAmap] & 
\Elsem{\clocksetA} & \eqdef \Elsem\Delta[\cpair{\clocksetAmap\circ \p}{\q}] 
\end{align*}

\begin{proposition} \label{prop:Usem:welldef}
 The objects $\Usem{\clocksetA}$ and $\Elsem{\clocksetA}$ are welldefined in the sense that
 they do not depend on the choice of $\Delta$ or surjection inducing $\clocksetAmap$. Moreover,
 if $\rho : \Gamma' \to \Gamma$ then 
\begin{align*}
\Usem{\clocksetA\circ \rho} & = \Usem{\clocksetA}[\rho] & 
\Elsem{\clocksetA\circ \rho} & = \Elsem{\clocksetA}[\cpair{\rho\circ \p}{\q}]
\end{align*}
where $\{\kappa_1, \dots, \kappa_n\} [\rho] = \{\kappa_1[\rho], \dots, \kappa_n[\rho]\}$.
\end{proposition}

\begin{proof}
 If $\gamma \in \Gamma_{(\FSA, \delta)}$ the element $\clocksetAmap(\gamma)$ is a map $\Delta \to \FSA$. 
 By definition, $\Usem{\clocksetA} (\gamma) = \Usem\Delta(\clocksetAmap(\gamma))$
 is the set of small families $(X_\tau)_{\tau : (\clocksetAmap(\gamma)[\Delta], \restrict\delta{\clocksetAmap(\gamma)[\Delta]}) 
 \to (\FSB, \delta')}$. Since the map $\Delta \to \clocksetA$ is assumed surjective, 
 $\clocksetAmap(\gamma)[\Delta] = \{\kappa(\gamma) \mid \kappa\in \clocksetA\}$ and thus independent
 of the choice of $\Delta$ and surjection. Since $\Elsem{\clocksetA}(\gamma)(X) = X_i$ for $i$ the inclusion, 
 also $\Elsem{\clocksetA}$ is welldefined. For the last statement, note that $\clocksetAmap \circ \rho$ 
 is the map corresponding to the composition $\Delta \to \clocksetA \to \clocksetA[\rho]$, where the last
 of these maps $\kappa$ to $\kappa[\rho]$, and this map is surjective. Therefore, 
 $\Usem{\clocksetA\circ \rho}$ can be defined as $\Usem\Delta[\clocksetAmap\circ\rho] = \Usem{\clocksetA}[\rho]$.  
 The equality $\Elsem{\clocksetA\circ \rho} = \Elsem{\clocksetA}[\cpair{\rho\circ \p}{\q}]$ follows similarly. 
\end{proof}

The codes on universes will be defined below by constructing objects $A_\Delta$ in $\gr\Delta$
and families $B_\Delta$ over $A_\Delta$ indexed over $\Delta$ in such a way that whenever
$\clocksetA$ is as above the families $A_\clocksetA = A_\Delta[\clocksetAmap]$ and 
$B_\clocksetA = B_\Delta[\cpair{\clocksetAmap\circ \p}{\q}]$ are well defined, i.e., independent of choice of
$\Delta$ and surjection $\Delta \to \clocksetA$. In this case, if each $A_\Delta$ and $B_\Delta$ are 
invariant under clock introduction, by Lemma~\ref{lem:Ufam:class} there is a unique code 
$\code{B_\Delta} : A_\Delta \to \Usem\Delta$ in $\gr\Delta$ 
such that $\Elsem\Delta[\code{B_\Delta}] = B_\Delta$. In this situation we would like to define 
\begin{align*}
 \code{B_{\clocksetA}} & = \code{B_\Delta}[\clocksetAmap] : A_{\clocksetA} \to \Usem{\clocksetA}
\end{align*}
as a map in the category of presheaves over the elements of $\Gamma$.

\begin{lemma}\label{lem:codes:welldef}
 In the situation described above, the map $\code{B_{\clocksetA}}$ is well defined, i.e., independent
 of the choice of $\Delta$ and surjection $\Delta \to \clocksetA$. Moreover, 
 $\Elsem{\clocksetA}[\code{B_{\clocksetA}}] = B_{\clocksetA}$, and if $\rho : \Gamma' \to \Gamma$
 then $\code{B_{\clocksetA\circ \rho}} = \code{B_{\clocksetA}}[\rho]$
\end{lemma}

\begin{proof}
 Suppose we are given two different surjections $f : \Delta \to \clocksetA$ and $f' : \Delta' \to \clocksetA$ 
 inducing $\clocksetAmap : \Gamma \to \Cl^\Delta$ and $\clocksetAmap' : \Gamma \to \Cl^{\Delta'}$. We will
 assume there is an surjection $g : \Delta \to \Delta'$ such that $f'g=f$, otherwise apply the argument to each of the two maps
 in the span of projections $\Delta \leftarrow \Delta \times \Delta' \to \Delta'$. Note that the projections are always surjective,
 since $\Delta$ is empty iff $\clocksetA$ is empty iff $\Delta'$ is empty. Since $f'g = f$ also 
 $\Cl^g\circ \clocksetAmap' = \clocksetAmap$. 

 We first prove that $A_{\Delta}[\Cl^g] = A_{\Delta'}$. For this, observe that there is a family 
 \[
 \Phi = \{ \ev_\kappa : \Cl^{\Delta'} \to \Cl \mid \kappa \in \Delta'\}
 \]
 and a surjection $\Delta' \to \Phi$ mapping $\kappa$ to $\ev_\kappa$. The induced map $\Cl^{\Delta'} \to \Cl^{\Delta'}$
 is the identity. There is also a map $\Delta \to \Phi$ mapping $\kappa$ to $\ev_{g(\kappa)}$. Since $g$ is surjective, also
 this is surjective, and induces $\Cl^g : \Cl^{\Delta'} \to \Cl^\Delta$. Thus by assumption 
\begin{align*}
 A_{\Delta}[\Cl^g] & = A_{\Delta'} & B_{\Delta}[\cpair{\Cl^g\circ \p}{\q}] = B_{\Delta'}
\end{align*}
In particular, these arguments apply to $\Usem{\Delta}$ and $\Elsem{\Delta}$ proving 
\begin{align*}
 \Usem{\Delta}[\Cl^g] & = \Usem{\Delta'} & \Elsem{\Delta}[\cpair{\Cl^g\circ \p}{\q}] = \Elsem{\Delta'}
\end{align*}
In the second of these equations $\Elsem{\Delta'}$  is considered a family of $\BaseCat$ over 
$\compr{\Cl^{\Delta'}}{\,\Usem{\Delta'}}$. Equivalently, $\Elsem{\Delta'}$  can be considered a
family of $\gr{\Delta'}$ over $\Usem{\Delta'}$, and $(-)[\Cl^g]$ a morphism of CwFs from $\gr{\Delta}$ 
to $\gr{\Delta'}$. From this latter point of view the second equation above is $\Elsem{\Delta}[\Cl^g] = \Elsem{\Delta'}$, and so 
\begin{align*}
\Elsem{\Delta'}[\code{B_\Delta}[\Cl^g]]
 = \Elsem{\Delta}[\Cl^g][\code{B_\Delta}[\Cl^g]]
 = (\Elsem{\Delta}[\code{B_\Delta}])[\Cl^g]
 = B_{\Delta}[\Cl^g] = B_{\Delta'}
\end{align*}
Thus, by the uniqueness statement of Lemma~\ref{lem:Ufam:class} $\code{B_\Delta}[\Cl^g] = \code{B_{\Delta'}}$. So, finally
\begin{align*}
 \code{B_\Delta}[\clocksetAmap] = \code{B_\Delta}[\Cl^g\circ \clocksetAmap'] = \code{B_{\Delta'}}[\clocksetAmap']
\end{align*}
proving welldefinedness of $\code{B_\clocksetA}$. The equality $\Elsem{\clocksetA}[\code{B_{\clocksetA}}] = B_{\clocksetA}$
follows from the fact that $(-)[\clocksetAmap]$ induces a morphism of CwFs:
\begin{align*}
 \Elsem{\clocksetA}[\code{B_{\clocksetA}}] & = \Elsem{\Delta}[\code{B_{\Delta}}][\clocksetAmap] = B_\Delta[\clocksetAmap] = B_{\clocksetA}
\end{align*}
The last statement follows as in the proof of Proposition~\ref{prop:Usem:welldef}. 
\end{proof}

\subsection{Inclusions of universes}

We now show how to model inclusions of universes and 
the codes for type operations on universes described 
in Figure~\ref{fig:universes:dep:prod}. 

\begin{proposition} \label{prop:univ:incl}
 Suppose $\clocksetA \subseteq \clocksetB$ are sets of elements 
 of $\clockfam[\uniquemap_\Gamma]$, and 
 ${t \in \ElCwF{\Gamma}{\Usem{\clocksetA}}}$. There is an element 
 $\inel{\clocksetA}{\clocksetB} (t) \in \ElCwF{\Gamma}{\Usem{\clocksetB}}$
 such that  $\Elsem{\clocksetB}[\cpair{\id\Gamma}{\inel{\clocksetA}{\clocksetB} (t)}]
 = \Elsem{\clocksetA}(t)$. 
 If further $\clocksetB \subseteq  \clocksetC$ then 
  $\inel{\clocksetB}{\clocksetC} (\inel{\clocksetA}{\clocksetB} (t))
  = \inel{\clocksetA}{\clocksetC} (t)$. 
 Moreover, this construction commutes with reindexing in the sense that if $\rho : \Gamma'\to \Gamma$ then 
  $(\inel{\clocksetA}{\clocksetB} (t))[\rho] 
  = \inel{\clocksetA[\rho]}{\clocksetB[\rho]} (t[\rho])$.
\end{proposition}

\begin{proof}
Let $\clocksetBmap : \Gamma \to \Cl^{\Delta'}$ be induced by a given 
surjection $\Delta' \to \clocksetB$. Let $\Delta\subseteq \Delta'$ be the subset mapped to 
$\clocksetA$, and let $\clocksetAmap : \Gamma \to \Cl^{\Delta}$ be the 
map corresponding to the projection. There is a projection 
$\Clproj\Delta{\Delta'} : \Cl^{\Delta'} \to \Cl^{\Delta}$ and so 
$\Usem{\Delta}[\Clproj\Delta{\Delta'}]$ is an object of $\gr{\Delta'}$. Moreover
\begin{equation} \label{eq:pi:Delta:Delta'}
\Usem{\Delta}[\Clproj\Delta{\Delta'}][\clocksetBmap] = \Usem{\clocksetA}
\end{equation}
simply because $\Clproj\Delta{\Delta'} \circ \clocksetBmap = 
\clocksetAmap$. 

 Since $\Usem{\Delta}$ and $\Elsem{\Delta}$ are invariant under clock introduction, so are 
 $\Usem{\Delta}[\Clproj\Delta{\Delta'}]$ and the family
 $\Elsem{\Delta}[\cpair{\Clproj\Delta{\Delta'}\circ \p}{\q}]$. The latter is a family over 
 $\compr{\Cl^{\Delta'}}{\,\Usem{\Delta}[\Clproj\Delta{\Delta'}]}$ in $\BaseCat$, but can 
 be likewise considered a family over $\Usem{\Delta}[\Clproj\Delta{\Delta'}]$ in $\gr{\Delta'}$. 
 By Lemma~\ref{lem:Ufam:class} there is a unique map $\Clin{\Delta}{\Delta'}$ in $\gr{\Delta'}$ such that 
 $\Elsem{\Delta'}[\Clin{\Delta}{\Delta'}] = \Elsem{\Delta}[\cpair{\Clproj\Delta{\Delta'}\circ \p}{\q}]$.
 By (\ref{eq:pi:Delta:Delta'}) then 
 \[
 \Clin{\Delta}{\Delta'}[\clocksetBmap] : \Usem{\clocksetA} \to 
 \Usem{\clocksetB} 
 \]
 is a map between presheaves over the category of elements of $\Gamma$, and
 the above implies 
 \[
 \Elsem{\clocksetB}[\Clin{\Delta}{\Delta'}[\clocksetBmap]] = 
 \Elsem{\clocksetA}
 \]
 We now define 
 \[
  \inel{\clocksetA}{\clocksetB} (t) \eqdef 
  \Clin{\Delta}{\Delta'}[\clocksetBmap](t)
 \]
 Then
 \begin{align*}
  \Elsem{\clocksetB}[\cpair{\id\Gamma}{\inel{\clocksetA}{\clocksetB} (t)}]
  & = \Elsem{\clocksetB}[\cpair{\id\Gamma}{\Clin{\Delta}{\Delta'}[\clocksetBmap](t)}] \\
  & = \Elsem{\clocksetA}[\cpair{\id{\Gamma}}t]
 \end{align*}
 
 The statement on composition of these inclusions follow from the uniqueness statement of Lemma~\ref{lem:Ufam:class}.
 The element $\inel{\clocksetA}{\clocksetB} (t)$ as defined above can be proved independent of the choice
 of $\Delta'$ using a slight generalisation of Lemma~\ref{lem:codes:welldef}, but we omit the argument here. Similar arguments
 can also show that it commutes with reindexing. 
\end{proof}

\subsection{Codes for basic type constructors}

The codes for $\Pi$ and $\Sigma$-types are modelled as morphisms with domain 
\[
\Ufam{\clocksetA} = 
\SigmaSem{\Usem{\clocksetA}}{\Elsem{\clocksetA} \to \Usem{\clocksetA}[\p]}
\]
The family $\Ufam{\clocksetA}$
classifies $\Usem{\clocksetA}$-small families over $\Usem{\clocksetA}$-small
objects in a sense that we now explain.

First note that there is a family 
\[
\ElfamOne{\clocksetA} \eqdef
\Elsem{\clocksetA}[\cpair{\p}{\pi_1(\q)}] \in \Fam{\compr{\Gamma}{\,\Ufam{\clocksetA}}}\] 
and an element 
\[\ev(\pi_2(\q)[\p], \q) \in\ElCwF{\compr{\Gamma}{\,\compr{\Ufam{\clocksetA}}{\ElfamOne{\clocksetA}}}}{\Usem{\clocksetA}[\p\p]},\] 
where $\pi_1, \pi_2$ are the projections out of the $\Sigma$-type. So
\[
\ElfamTwo{\clocksetA} \eqdef
\Elsem{\clocksetA}[\cpair{\p\p}{\ev(\pi_2(\q)[\p], \q)}] \in \Fam{\compr{\Gamma}{\,\compr{\Ufam{\clocksetA}}{\ElfamOne{\clocksetA}}}}
\]

%

Suppose now
$A \in \Fam\Gamma$  and $B \in \Fam{\compr\Gamma A}$ are $\Usem{\clocksetA}$-
small in the sense that there are $\code A$ and $\code B$ satisfying
\begin{align*}
 \code A & \in \ElCwF{\Gamma}{\Usem{\clocksetA}}
 & \Elsem{\clocksetA}[\cpair{\id\Gamma}{\code A}] & = A \\
 \code B & \in \ElCwF{\compr\Gamma A}{\Usem{\clocksetA}[\p]} 
 & \Elsem{\clocksetA}[\cpair{\p}{\code B}] & = B
\end{align*}
Now, 
$\lambdael{\code B} \in \ElCwF{\Gamma}{A \to \Usem{\clocksetA}}$ 
and so
\[
\pair{\code A}{\lambdael{\code B}} \in \ElCwF{\Gamma}{\Ufam{\clocksetA}}
\]
Then 
\begin{align*}
\ElfamOne{\clocksetA}[\cpair{\id\Gamma}{\pair{\code A}{\lambdael{\code B}}}] & = 
\Elsem{\clocksetA}[\cpair{\p}{\pi_1(\q)}][\cpair{\id\Gamma}{\pair{\code A}{\lambdael{\code B}}}] \\
& = \Elsem{\clocksetA}[\cpair{\id\Gamma}{\code A}] \\ & = A 
\end{align*}
and 
\begin{align*}
\ElfamTwo{\clocksetA}[\cpair{\cpair{\p}{\pair{\code A}{\lambdael{\code B}}\p}}{\q}] 
& =
\Elsem{\clocksetA}[\cpair{\p\p}{\ev(\pi_2(\q)[\p], \q)}][\cpair{\cpair{\p}{\pair{\code A}{\lambdael{\code B}}\p}}{\q}] \\
& = \Elsem{\clocksetA}[\cpair{\p}{\ev(\lambdael{\code B}[\p], \q)}]  \\
& = \Elsem{\clocksetA}[\cpair{\p}{\code B}] \\
& = B 
\end{align*}

\begin{proposition} \label{prop:codes:Pi:Sigma}
 Suppose $A, \code A$, $B$ and $\code B$ are as above. There are elements 
\begin{align*}
   \Piel{\clocksetA}{\code A}{\code B}  & \in \ElCwF{\Gamma}{\Usem{\clocksetA}} &
   \Sigmael{\clocksetA}{\code A}{\code B}  & \in \ElCwF{\Gamma}{\Usem{\clocksetA}} 
\end{align*}
 such that 
\begin{align*}
 \Elsem{\clocksetA}\left[\cpair{\id\Gamma}{\Piel{\clocksetA}{\code A}{\code B}}\right] 
 & = \PiSem AB & 
 \Elsem{\clocksetA}\left[\cpair{\id\Gamma}{\Sigmael{\clocksetA}{\code A}{\code B}}\right] 
 & = \SigmaSem AB
\end{align*}
Moreover, if $\rho : \Gamma' \to \Gamma$ then 
\begin{align*}
 (\Piel{\clocksetA}{\code A}{\code B})[\rho] 
 & = \Piel{\clocksetA[\rho]}{\code A[\rho]}{\code B[\rho]} \\
 (\Sigmael{\clocksetA}{\code A}{\code B})[\rho] 
 & = \Sigmael{\clocksetA[\rho]}{\code A[\rho]}{\code B[\rho]}  \\
 \inel{\clocksetA}{\clocksetB} (\Piel{\clocksetA}{\code A}{\code B}) 
 & = \Piel{\clocksetB}{\inel{\clocksetA}{\clocksetB}(\code A)}{\inel{\clocksetA}{\clocksetB} (\code B)} \\
 \inel{\clocksetA}{\clocksetB} (\Sigmael{\clocksetA}{\code A}{\code B}) 
 & = \Sigmael{\clocksetB}{\inel{\clocksetA}{\clocksetB}(\code A)}{\inel{\clocksetA}{\clocksetB} (\code B)}
\end{align*}
\end{proposition}

\begin{proof}
 First note 
 that since $\Usem\Delta$ and $\Elsem\Delta$ are invariant under clock introduction,
 by the closure of these under $\Pi$, $\Sigma$ and 
 reindexing (Corollary~\ref{cor:orthogonality:closed}), 
 so is $\Ufam\Delta$. By a similar argument, 
 also the families $\ElfamOne{\Delta}$, $\ElfamTwo{\Delta}$
 and $\PiSem{\ElfamOne{\Delta}}{\ElfamTwo{\Delta}}$ are 
 invariant under clock introduction. By Lemma~\ref{lem:Ufam:class} 
 there is a unique morphism $\code{\Pi_\Delta} : \Ufam\Delta \to \Usem\Delta$ in $\gr\Delta$ such that 
 \[ \PiSem{\ElfamOne{\Delta}}{\ElfamTwo{\Delta}} = \Elsem\Delta[\code{\Pi_\Delta}]
 \]
 The element
 \[
  \Piel{\clocksetA}{\code A}{\code B} \eqdef 
  \code{\Pi_\Delta}[\clocksetAmap](\pair{\code A}{\lambdael{\code B}})
 \]
 is then welldefined by Lemma~\ref{lem:codes:welldef} and satisfies
\begin{align*}
 \Elsem{\clocksetA}\left[\cpair{\id\Gamma}{\Piel{\clocksetA}{\code A}{\code B}}\right] 
 & = \PiSem{\ElfamOne{\clocksetA}}{\ElfamTwo{\clocksetA}}
 [\cpair{\id\Gamma}{\pair{\code A}{\lambdael{\code B}}}] \\
  & = \PiSem AB
\end{align*}

For the last statement, note that the map $\Clin{\Delta}{\Delta'} : \Usem\Delta[\Clproj\Delta{\Delta'}] \to \Usem{\Delta'}$ 
in $\gr{\Delta'}$ from the proof of Proposition~\ref{prop:univ:incl} induces a map
$\Ufam\Delta[\Clproj\Delta{\Delta'}] \to \Ufam{\Delta'}$ mapping $\pair{A}{B}$ to $\pair{\Clin{\Delta}{\Delta'}(A)}{\Clin{\Delta}{\Delta'} \circ B}$.
This makes the following diagram commute
\[
\begin{diagram}
 \Ufam\Delta[\Clproj\Delta{\Delta'}] \ar{r} \ar[swap]{d}{\code{\Pi_\Delta}[\Clproj\Delta{\Delta'}]} & \Ufam{\Delta'} \ar{d}{\code{\Pi_{\Delta'}}} \\
 \Usem\Delta[\Clproj\Delta{\Delta'}] \ar{r}{\Clin{\Delta}{\Delta'}} & \Usem{\Delta'} 
\end{diagram}
\]
by the uniqueness statement of Lemma~\ref{lem:Ufam:class} because reindexing $\Elsem{\Delta'}$ along either direction gives
$\PiSem{\ElfamOne{\Delta}[\Clproj\Delta{\Delta'}]}{\ElfamTwo{\Delta}[\cpair{\Clproj\Delta{\Delta'}\circ\p}{\q}]}$. From this the final
statement follows. 
\end{proof}

\subsection{Universal quantification over clocks}

We now describe the codes for universal quantification over clocks. Even though clock quantification is modelled
using $\Pi$-types, this is not a special case of Proposition~\ref{prop:codes:Pi:Sigma}, since on the level of codes, 
clock quantification involves a change of universe. 
\begin{proposition}
Suppose $A \in \Fam{\compr\Gamma{\clockfam[\uniquemap]}}$, and that
$\clocksetA$ is a set of elements of $\clockfam[\uniquemap_\Gamma]$. Write $\clocksetA[\p], \q$ for
the union of the set $\clocksetA[\p]$ and $\q$, and suppose 
\[
 \code A \in \ElCwF{\compr{\Gamma}{\clockfam[\uniquemap]}}{\Usem{\clocksetA[\p], \q}}
\]
is such that $A = \Elsem{\clocksetA[\p], \q}[\cpair{\id{\compr\Gamma{\clockfam[\uniquemap]}}}{\code A}]$. 
There is an element 
$\forallel{\clocksetA}{\code A} \in \ElCwF{\Gamma}{\Usem{\clocksetA}}$
such that
\[
 \Elsem{\clocksetA}\left[\cpair{\id\Gamma}{\forallel{\clocksetA}{\code A}}\right] 
 = \PiSem{\clockfam[\uniquemap]}{A}
\] 
Moreover, if $\rho : \Gamma' \to \Gamma$ then 
\[
  (\forallel{\clocksetA}{\code A})[\rho] = \forallel{\clocksetA[\rho]}{\code A[\cpair{\rho \circ \p}{\q}]}
\]
and if $\clocksetA\subseteq\clocksetB$ then
\[
\inel{\clocksetA}{\clocksetB}(\forallel{\clocksetA}{\code A}) 
= \forallel{\clocksetB}{\inel{(\clocksetA[\p], \q)}{(\clocksetB[\p], \q)}(\code A)}
\]
\end{proposition}

\begin{proof}
The universe $\Usem{\Delta, \kappa}$ is an object in $\gr{\Delta, \kappa}$, which means 
that it is a family over $\Cl^{\Delta, \kappa}$. Abusing notation slightly, write 
$\Usem{\Delta, \q}$ for the family over $\compr{\Cl^{\Delta}}{\clockfam[\uniquemap]}$
obtained by reindexing $\Usem{\Delta, \kappa}$ along the isomorphism 
$\Cl^{\Delta, \kappa} \iso \compr{\Cl^{\Delta}}{\clockfam[\uniquemap]}$
and write $\Elsem{\Delta, \q}$ for the result of reindexing the family 
$\Elsem{\Delta, \kappa}$ along the same map. Note that 
$\Usem{\Delta, \q}[\cpair{\clocksetAmap \circ \p}{\q}] = 
\Usem{\clocksetA[\p], \q}$ and 
$\Elsem{\Delta, \q}[\cpair{\cpair{\clocksetAmap \circ \p\p}{\q[\p]}}{\q}] = 
\Elsem{\clocksetA[\p], \q}$.

We will now construct the generic clock quantified family over  
$\compr{\Cl^\Delta}{\PiSem{\clockfam[\uniquemap]}{\Usem{\Delta,\q}}}$ and construct the 
$\forallel{\clocksetA}{\code A}$ using Lemma~\ref{lem:Ufam:class}. First observe that
\[
 \ev(\q[\p], \q) \in 
 \ElCwF{\compr{\compr{\Cl^\Delta}{\PiSem{\clockfam[\uniquemap]}{\Usem{\Delta,\q}}}}{\clockfam[\uniquemap]}}
 {\Usem{\Delta,\q}[\cpair{\p\p}{\q}]}
\]
and so
\[
 \Elsem{\Delta, \q}[\cpair{\cpair{\p\p}{\q}}{\ev(\q[\p], \q)}] \in \Fam{\compr{\compr{\Cl^\Delta}{\PiSem{\clockfam[\uniquemap]}{\Usem{\Delta,\q}}}}{\clockfam[\uniquemap]}}
\]
and 
\[
 \PiSem{\clockfam[\uniquemap]}{\Elsem{\Delta, \q}[\cpair{\cpair{\p\p}{\q}}{\ev(\q[\p], \q)}]} \in \Fam{\compr{\Cl^\Delta}{\PiSem{\clockfam[\uniquemap]}{\Usem{\Delta,\q}}}}
\]
Since $\Usem{\Delta,\kappa}$ and $\Elsem{\Delta,\kappa}$ are invariant under clock introduction and this notion is closed
under reindexing and $\Pi$-types, also $\PiSem{\clockfam[\uniquemap]}{\Usem{\Delta,\q}}$ and
$\PiSem{\clockfam[\uniquemap]}{\Elsem{\Delta, \q}[\cpair{\cpair{\p\p}{\q}}{\ev(\q[\p], \q)}]}$ are invariant under clock introduction,
and so by Lemma~\ref{lem:Ufam:class} there is a map 
\[
\code{\forall^{\Delta}} : \PiSem{\clockfam[\uniquemap]}{\Usem{\Delta,\q}} \to \Usem{\Delta}
\]
in $\gr\Delta$ such that 
\[
\PiSem{\clockfam[\uniquemap]}{\Elsem{\Delta, \q}[\cpair{\cpair{\p\p}{\q}}{\ev(\q[\p], \q)}]} = \Elsem{\Delta}[\code{\forall^{\Delta}}]
\]
By Lemma~\ref{lem:codes:welldef} the element 
\[
\forallel{\clocksetA}{\code A} \eqdef \code{\forall^{\Delta}}[\clocksetAmap](\lambdael{\code A})
\] 
is welldefined and satifies
\begin{align*}
 \Elsem{\clocksetA}\left[\cpair{\id\Gamma}{\forallel{\clocksetA}{\code A}}\right] 
 & = \PiSem{\clockfam[\uniquemap]}{\Elsem{\clocksetA[\p], \q}[\cpair{\cpair{\p\p}{\q}}{\ev(\q[\p], \q)}]} [\cpair{\id{\Gamma}}{\lambdael{\code A}}] \\
 & = \PiSem{\clockfam[\uniquemap]}{\Elsem{\clocksetA[\p], \q}[\cpair{\cpair{\p\p}{\q}}{\ev(\q[\p], \q)}][\cpair{\cpair{\id{\Gamma}}{\lambdael{\code A}}\circ\p}{\q}]} \\
 & = \PiSem{\clockfam[\uniquemap]}{\Elsem{\clocksetA[\p], \q}[\cpair{\cpair{\p}{\q}}{\ev(\lambdael{\code A}[\p], \q)}]} \\
 & = \PiSem{\clockfam[\uniquemap]}{\Elsem{\clocksetA[\p], \q}[\cpair{\id{}}{\code A}]} \\
 & = \PiSem{\clockfam[\uniquemap]}{A} 
\end{align*}
This construction is clearly closed under reindexing, and the last statement can be proved similarly to the last statement of 
Proposition~\ref{prop:codes:Pi:Sigma}.
\end{proof}

\subsection{Codes for the later modalities}
\label{sec:codes-for-later-modalities}

\begin{proposition}
 There is a mapping associating $\kappa \in \clocksetA$ and  $t \in \ElCwF{\Gamma}{\laterfam{\kappa}\Usem{\clocksetA}}$ to  
 \[
 \laterel{\kappa}(t) \in  \ElCwF{\Gamma}{\Usem{\clocksetA}}
 \]
 such that if $\Gamma'$ is a telescope of length $m$ over $\Gamma$ and 
 $\code A \in \ElCwF{\compr\Gamma{\Gamma'}}{\Usem{\clocksetA[\p^m]}}$ and 
 $A = \Elsem{\clocksetA[\p^m]}[\cpair{\id{\compr{\Gamma}{\Gamma'}}}{\code A}]$  and 
 $\xi  \in \delayedEl{\Gamma}{\Gamma'}{\kappa}$
 then 
 \[
   \Elsem{\clocksetA}\left[\cpair{\id\Gamma}{\laterel{\kappa}(\nextel{\kappa} \xi . \code A)}\right] = 
   \laterfam{\kappa}\xi . A 
 \]
 Moreover, if $\rho : \Gamma'' \to \Gamma$
 then $\laterel{\kappa}(t)[\rho] = \laterel{\kappa[\rho]}(t[\rho])$, and if 
 $\clocksetA\subseteq\clocksetB$ then
\[
\inel{\clocksetA}{\clocksetB}(\laterel{\kappa}(t)) 
= \laterel{\kappa}(\nextel{\kappa}(t).\inel{\clocksetA[\p]}{\clocksetB[\p]}(\q))
\]
\end{proposition}

A few of the typings of the proposition need to be explained. The element $\nextel{\kappa} \xi . \code A$
is a priori an element of $\laterfam{\kappa}\xi . (\Usem{\clocksetA}[\p^m])$ but the latter family
equals $\laterfam{\kappa}\Usem{\clocksetA}$ and so $\laterel{\kappa}(\nextel{\kappa} \xi . \code A)$
is well formed. In the last equation, $t$ is considered a delayed sequence of elements in 
$\delayedEl{\Gamma}{\Usem{\clocksetA}}{\kappa}$, and 
$\inel{\clocksetA[\p]}{\clocksetB[\p]}(\q) \in \ElCwF{\compr{\Gamma}{\,\Usem{\clocksetA}}}{\Usem\clocksetB[\p]}$ and so 
\[
 \nextel{\kappa}(t).\inel{\clocksetA[\p]}{\clocksetB[\p]}(\q) \in \ElCwF{\Gamma}{\laterfam{\kappa}(t).\,\Usem\clocksetB[\p]} 
 =  \ElCwF{\Gamma}{\laterfam{\kappa}\Usem\clocksetB}
\]
making the right hand side of the final equation well formed. 

\begin{proof}
 Suppose note that any $\hat\kappa \in \Delta$ defines an element
 $\hat\kappa \in \ElCwF{\Cl^\Delta}{\clockfam[\uniquemap]}$ essentially by projection. Since
 $\q \in \ElCwF{\compr{\Cl^\Delta}{\laterfam{\hat\kappa} \Usem\Delta}}{(\laterfam{\hat\kappa}\Usem\Delta)[\p]}$
 it defines a delayed sequence of elements 
 $(q) \in \delayedEl{\compr{\Cl^\Delta}{\laterfam{\hat\kappa} \Usem\Delta}}{\Usem\Delta[\p]}{\hat\kappa[\p]}$. 
 Since moreover, $\Elsem\Delta[\cpair{\p\p}{\q}]$ is a family in
 $\Fam{\compr{\compr{\Cl^\Delta}{\laterfam{\hat\kappa} \Usem\Delta}}{\,\Usem\Delta[\p]}}$ we can define
 \[
  \laterfam{\hat\kappa[\p]}(\q) . (\Elsem\Delta[\cpair{\p\p}{\q}]) \in \Fam{\compr{\Cl^\Delta}{\laterfam{\hat\kappa} \Usem\Delta}}
 \]
 By Proposition~\ref{prop:laterfam:inv:clock:intro}, both $\laterfam{\hat\kappa} \Usem\Delta$ and 
 $\laterfam{\hat\kappa[\p]}(\q) . (\Elsem\Delta[\cpair{\p\p}{\q}])$ are invariant under clock introduction and so
 by Lemma~\ref{lem:Ufam:class} there is a morphism $\code{\laterfam{\hat\kappa}} : \laterfam{\hat\kappa} \Usem\Delta \to \Usem\Delta$
 in $\gr\Delta$ such that $\Elsem\Delta[\code{\laterfam{\hat\kappa}}] = \laterfam{\hat\kappa[\p]}(\q) . (\Elsem\Delta[\cpair{\p\p}{\q}])$.
 Using this we define 
 \[
 \laterel{\kappa}(t) \eqdef \code{\laterfam{\hat\kappa}}[\clocksetAmap] (t)
 \]
 where $\hat\kappa$ is an element in $\Delta$ mapped to $\kappa$. This can be proved independent of choice of $\Delta$
 and surjection $\Delta \to \clocksetAmap$ and $\hat\kappa$ using arguments as in the proof of 
 Lemma~\ref{lem:codes:welldef}. Then
\begin{align*}
 \Elsem{\clocksetA}\left[\cpair{\id\Gamma}{\laterel{\kappa}(\nextel{\kappa} \xi . \code A)}\right] 
 & = (\laterfam{\kappa[\p]}(\q) . \Elsem{\clocksetA}[\cpair{\p\p}{\q}])
 \left[\cpair{\id\Gamma}{\nextel{\kappa} \xi . \code A}\right] \\
 & = \laterfam{\kappa}(\nextel{\kappa} \xi . \code A) . (\Elsem{\clocksetA}[\cpair{\p\p}{\q}]) \\
 & = \laterfam{\kappa}\xi[\nextel{\kappa} \xi . \code A] . (\Elsem{\clocksetA}[\cpair{\p\p}{\q}][\cpair{\p^{m+1}}{\q}])
\end{align*}
where in the last step $m$ is the length of $\xi$. Now by Theorem~\ref{thm:later:del:subst:rules}.\ref{item:beta}, the latter equals
\begin{align*}
 \laterfam{\kappa}\xi . (\Elsem{\clocksetA}[\cpair{\p\p}{\q}][\cpair{\p^{m}}{\code A}])
 & = \laterfam{\kappa}\xi . (\Elsem{\clocksetA}[\cpair{\p^{m+1}}{\code A}]) \\
 & = \laterfam{\kappa}\xi . (\Elsem{\clocksetA[\p^m]}[\cpair{\id{}}{\code A}]) \\
 & = \laterfam{\kappa}\xi .A
\end{align*}
For the final statement, can be proved using the uniqueness statement of Lemma~\ref{lem:Ufam:class}.
\end{proof}

\section{Interpreting syntax}
\label{sec:interp:syntax}

The previous sections define the semantic structure of the model corresponding to each of the 
constructions of \gdtt. One can use this to define an interpretation of the syntax into the model,
as we briefly sketch here. As is well known, defining interpretation of dependent type theories is
not a simple procedure. In particular, the proof of welldefinedness of the interpretation can not
be separated from the proof of soundness. Here we follow the approach of 
\citeasnoun{Hofmann:syntax-and-semantics}, which 
first defines an interpretation of (pre-) contexts, types and term as a partial function, then proves 
that this function is defined on all wellformed judgements. To define the partial interpretation 
function, syntax must be annotated with typing information, meaning that the syntax interpreted
is not quite the syntax usually presented for dependent type theory. For example, $\lambda$-abstractions
must be annotated with not just the type of the variable being abstracted, but also with the target
type of the function created (which is a dependent family $(x.A)$). Likewise, application is annotated
both with the domain type and with the dependent codomain type. 

Definedness of the interpretation of well formed judgements is then proved by induction on the 
structure of judgements. This must be done simultaneously with the proof of soundness of the 
interpretation and with the proof of a substitution lemma. We now sketch how each of these 
ingredients must be adapted to interpret \gdtt. 

The annotation of terms and types must be extended to the new constructions. Universal quantification
over clocks is interpreted as a $\Pi$-type, and the annotations must therefore be similar to those of $\Pi$-types.
Terms like $\next{\kappa}$ and $\fix\kappa$ must be annotated with the type at which they are applied.
Recall from Section~\ref{sec:prev} that $\prev$ is compiled away in an initial step using $\force$. 
The constant $\force$ must be annotated with the dependent type ($\kappa . A$) at which it is
applied. Type operations on the universe must be annotated with the context $\Delta$ at which they are 
applied. The type constructor $\later\kappa\xi . A$ must be annotated with the types in the telescope
and likewise for $\next\kappa\xi.t$. A notion of pre- delayed substitutions must be defined and these
must be (partially) interpreted as delayed sequences of elements.

Once the partial interpretation function has been interpreted, the welldefinedness of the interpretation
of wellformed judgements must be proved by induction on judgements simultaneously with soundness
and a substitution lemma. To this sequence of lemmas must be added the statement that the 
interpretation of any type is invariant under clock introduction.

The substitution lemma is mostly standard. In particular, the notation of substitution 
between contexts can be defined essentially as usual
\begin{mathpar}
\inferrule*{\,}{\cdot : \Gamma \to \empctx} \and
\inferrule*{\rho : \Gamma \to \Gamma' \\ \hastype{\Gamma}{t}{A\rho}}
{\rho[x \mapsto t] :\Gamma \to \Gamma', x : A} \and
\inferrule*{\rho : \Gamma \to \Gamma' \\ \hastype{\Gamma}{\kappa'}{\clocktype}}
{\rho[\kappa \mapsto \kappa'] : \Gamma \to \Gamma', \kappa : \clocktype} 
\end{mathpar}
and likewise the notion of substitution is defined in the standard way. Note in particular that
this means that $(\univ\Delta)\rho = \univ{\Delta\rho}$. Substitution on a delayed substitution $\xi$ is defined
by distributing the interpretation over the terms in $\xi$. The substitution lemma is as follows.

\begin{lemma}
  \label{lem:substitution}
  If $\rho : \Gamma' \to \Gamma$ is a substitution, then 
  \begin{itemize}
  \item if $\wftype{\Gamma}{A}$ also $\wftype{\Gamma'}{A\rho}$ and $\den{\Gamma' \vdash A\rho} 
  = \den{\Gamma \vdash A}[\den{\rho}]$. 
  \item if $\hastype{\Gamma}t{A}$ also $\hastype{\Gamma'}{t\rho}{A\rho}$ and 
  $\den{\Gamma' \vdash t\rho} = \den{\Gamma \vdash t}[\den\rho]$
  \item if $\dsubst{\kappa}{\xi}{\Gamma}{\Gamma''}$ then also $\dsubst{\kappa\rho}{\xi\rho}{\Gamma'}{\Gamma''\rho}$
  and $\den{\xi\rho} = \den\xi[\den\rho]$.  
  \end{itemize}
\end{lemma}

\section{Recovering the categories \texorpdfstring{$\grold{\Delta}$}{GR(Delta)}}
\label{sec:discussion}

In this final section we discuss the relation to the family of categories
$\grold\Delta$ defined in previous work by the authors~\cite{Bizjak-Moegelberg:clocks-model}.
As mentioned in the introduction, this gives a model of guarded recursion with multiple clocks up 
to a coherence problem. We first recall the definition of the categories $\grold\Delta$ (note that the notation
for this differs from the $\gr\Delta$ used in the paper only by the choice of font). 

For a finite set of clock variables $\Delta$ the category $\grold{\Delta}$ is the category of presheaves on the poset 
$\pindex{\Delta}$.
The elements of this poset are pairs $(E, \delta)$ where $E$ is an equivalence relation on $\Delta$ and $\delta : \Delta \to \NN$ is a function which respects the equivalence relation $E$.
The order on $\pindex{\Delta}$ is defined so that $(E, \delta) \leq (E', \delta')$ if $E$ is coarser than $E'$ 
(i.e., $E'\subseteq E$ as subsets of $\Delta \times \Delta$) and $\delta$ is pointwise less than $\delta'$.
The idea behind this poset is that $\delta$ records how much time is left on each clock, and the equivalence relation $E$ states which clocks are identified.
The order is defined so that we can pass to a state where there is less time available on each clock, but we can also identify different clocks, i.e., make the equivalence relation coarser.

The intention of the categories $\grold\Delta$ is that types and terms in clock variable context $\Delta$ should be modelled in $\grold\Delta$.
In the present paper, the corresponding fragment is modelled in $\gr\Delta$ with the restriction that families 
must be invariant under clock introduction. Thus the next theorem states that the two models are equivalent.
\begin{theorem}
  \label{thm:equivalence-of-models}
  Let $\Delta$ be a finite set of clocks.
  The full subcategory of $\gr\Delta$ on objects invariant under clock introduction is equivalent to the category $\grold{\Delta}$.
\end{theorem}

\begin{proofsketch}
  \newcommand{\site}[1]{\mathcal{R}\left(#1\right)}
  \newcommand{\qindex}[1]{\mathcal{S}\left(#1\right)}
  Recall that $\gr\Delta$ is defined as the category of covariant presheaves on the category of elements of $\Cl^\Delta$, for which we
  write $\site{\Delta}$ in this proof.
%
  The indexing poset $\pindex{\Delta}$
  is equivalent to the preorder $\qindex{\Delta}^{\text{op}}$ where $\qindex{\Delta}$ is the full subcategory of $\site{\Delta}$ on those objects $(\FSA,\delta, f)$ where $f$ is surjective.
  Indeed, this equivalence follows from the fact that every function $f$ on $\Delta$ determines an equivalence relation on $\Delta$, and every equivalence relation $E$ on $\Delta$ gives rise to the surjective quotient function $q : \Delta \to \Delta/E$.
  Straightforward calculations show this extends to the claimed equivalence of the poset $\pindex{\Delta}$ and the preorder $\qindex{\Delta}^{\text{op}}$.

  Thus we have that $\grold{\Delta}$ is equivalent to the category of covariant presheaves on $\qindex{\Delta}$.
  By definition there is an inclusion functor $i : \qindex{\Delta} \to \site{\Delta}$ which gives rise, by precomposition, to a functor $i^* : \gr\Delta \to \grold{\Delta}$.
  Moreover, there is a functor $g : \site{\Delta} \to \qindex{\Delta}$ which maps $(\FSA,\delta,f)$ to $\left(f[\Delta], \restrict{\delta}{f[\Delta]},f\right)$.
  This functor gives rise to a functor $g^* : \grold{\Delta} \to \gr\Delta$.
  It is easy to see $g \comp i = \id{}$ and that there is a natural transformation $\eps : i \comp g \to \id{}$ whose component at $(\FSA,\delta,f)$ is given by the inclusion $f[\Delta] \to \FSA$.
  These transformations define an adjunction $i \dashv g$.

  Thus $g^* \dashv i^*$ as well; first from $g \comp i = \id{}$ we have $i^* \comp g^* = \id{}$, and so the unit $\eta^*$ of the adjunction is the identity natural transformation, and, second, from the counit of the adjunction $i \dashv g$ we define the counit $\eps^*$ of the adjunction $g^* \dashv i^*$ pointwise, as in
  \begin{align*}
    \left(\eps^*_X\right)_{(\FSA,\delta,f)} = X\left(\eps_{(\FSA,\delta,f)}\right).
  \end{align*}
  
  It is standard that an adjunction restricts to an equivalence of full subcategories $\CC$ of $\grold{\Delta}$ and $\DD$ of $\gr\Delta$ on objects where the unit and the counit are isomorphisms, respectively.
  Because the unit $\eta^*$ of the adjunction $g^* \dashv i^*$ is an isomorphism the category $\CC$ is $\grold{\Delta}$.
  
  The category $\DD$ on the other hand is the category of those objects $X \in \gr\Delta$ where for every $(\FSA,\delta,f) \in \site{\Delta}$ the component of the counit
  \begin{align*}
    \left(\eps^*_X\right)_{(\FSA,\delta,f)} = X\left(\eps_{(\FSA,\delta,f)}\right) = X(\iota)
  \end{align*}
  where $\iota : \left(f[\Delta], \restrict{\delta}{f[\Delta]},f\right) \to (\FSA,\delta,f)$ is the inclusion, is an isomorphism.
  By Lemma~\ref{lem:fibrations-iff-restrictions-pullback} this holds precisely when the object $X$ is invariant under clock introduction.
  Hence, the adjunction $g^* \dashv i^*$ restricts to the equivalence of $\grold{\Delta}$ and the full subcategory of $\gr\Delta$ on objects invariant under clock introduction.
\end{proofsketch}

Notice, however, that the categories in Theorem~\ref{thm:equivalence-of-models} are \emph{not} isomorphic.
This is the key to achieving preservation of structure, chiefly dependent products, in the present model, up to \emph{equality}, as opposed to only up to isomorphism, as in the previous model~\cite{Bizjak-Moegelberg:clocks-model}.

\section*{Acknowledgements}

We thank Patrick Bahr, Lars Birkedal, Hans Bugge Grathwohl and Bassel Mannaa for helpful discussions.
We thank the anonymous reviewers for helpful suggestions which led to a major revision significantly improving the paper. 
Bizjak was supported by the ModuRes Sapere Aude Advanced Grant from The Danish Council for Independent Research for the Natural Sciences (FNU).
M{\o}gelberg was supported by a research grant (13156) from VILLUM FONDEN and DFF-Research Project 1 Grant no.
4002-00442, from The Danish Council for Independent Research for the Natural Sciences (FNU).

\bibliography{paper}

\end{document}